%% file: main.tex
\documentclass[10pt,a4paper]{article}
\usepackage{a4wide}

\usepackage{amsmath, amsthm, amssymb}
\usepackage[utf8]{inputenc}
\usepackage[T1]{fontenc}
\usepackage{sansmath}
\usepackage{stackengine}
\usepackage{wrapfig}

\usepackage[bbgreekl]{mathbbol}
\input{macros}

\date{}

\title{Differential Logical Relations\protect\\ \vspace{-5pt} {\Large Part I: The Simply-Typed Case}}

\author{Ugo Dal Lago \and Francesco Gavazzo \and Akira Yoshimizu}

\begin{document}

\maketitle 

\begin{abstract}
  We introduce a new form of logical relation which, in the spirit of
  metric relations, allows us to assign each pair of programs a quantity
  measuring their distance, rather than a boolean value standing for
  their being equivalent. The novelty of differential logical
  relations consists in measuring the distance between terms not
  (necessarily) by a numerical value, but by a mathematical object
  which somehow reflects the interactive complexity, i.e. the type, of
  the compared terms. We exemplify this concept in the simply-typed
  lambda-calculus, and show a form of soundness theorem. We also see
  how ordinary logical relations and metric relations can be seen as
  instances of differential logical relations. Finally, we show that
  differential logical relations can be organised in a cartesian
  closed category, contrarily to metric relations, which are
  well-known \emph{not} to have such a structure, but only that of a
  monoidal closed category.
\end{abstract}

\section{Introduction}

Modern software systems tend to be heterogeneous and complex, and this is
\begin{wrapfigure}{R}{.22\textwidth}
  \fbox{
    \begin{minipage}{.2\textwidth}
      \centering
    \includegraphics[scale=1.0]{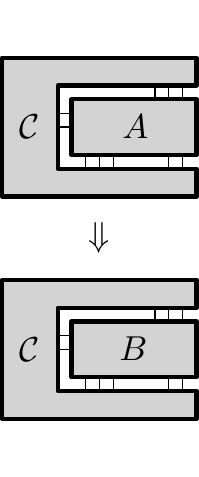}
  \end{minipage}}
  \caption{Replacing $A$ with $B$.}\label{fig:systemenv}
\end{wrapfigure}
reflected in the analysis methodologies we use to tame their complexity.
Indeed, in many cases the only way to go is to make use of compositional
kinds of analysis, in which \emph{parts} of a large system can be analysed
in isolation, without having to care about the rest of the system, the
\emph{environment}. As an example, one could consider a component
$A$ and replace it with another, e.g. more efficient component $B$ without
looking at the context $\mathcal{C}$ in which $A$ and $B$ are supposed to
operate, see Figure~\ref{fig:systemenv}. Of course, for this program transformation
to be safe, $A$ should be \emph{equivalent} to $B$ or, at least, $B$ should
be a \emph{refinement} of $A$.

Program equivalences and refinements, indeed, are the cruxes of program
semantics, and have been investigated in many different programming
paradigms. When programs have an interactive behaviour, like in
concurrent or higher-order languages, even \emph{defining} a notion of
program equivalence is not trivial, while coming out with handy
methodologies for \emph{proving} concrete programs to be equivalent
can be quite challenging, and has been one of the major research
topics in programming language theory, stimulating the development of
techniques like logical relations~\cite{Plotkin1973,Mitchell1996},
applicative bisimilarity~\cite{Abramsky/RTFP/1990}, and to some extent denotational
semantics~\cite{Scott,ScottStrachey} itself.

Coming back to our example, may we say anything about the case in
which $A$ and $B$ are \emph{not} equivalent, although behaving very
similarly? Is there anything classic program semantics can say
about this situation? Actually, the answer is negative: the program
transformation turning such an $A$ into $B$ cannot be justified,
simply because there is no guarantee about what the possible negative
effects that turning $A$ into $B$ could have on the overall system formed
by $\mathcal{C}$ and $A$.  There
are, however, many cases in which program transformations like the one
we just described are indeed of interest, and thus desirable. Many
examples can be, for instance, drawn from the field of
\emph{approximate computing}~\cite{Mittal2016}, in which equivalence-breaking program
transformations are considered as beneficial \emph{provided} the
overall behaviour of the program is not affected too much
by the transformation, while its intensional behaviour, e.g. its
performance, is significantly improved.

\newcommand{\FUZZ}{\ensuremath{\mathsf{FUZZ}}}
One partial solution to the problem above consists in considering
program \emph{metrics} rather than program \emph{equivalences}. This
way, any pair of programs are dubbed being at a certain numerical
distance rather than being merely equivalent (or not). This, for
example, can be useful in the context of differential 
privacy~\cite{Pierce/DistanceMakesTypesGrowStronger/2010,DBLP:journals/siglog/BartheGHP16,Kosta/metrics-for-differential-privacy/2014}
and has also been studied in the realms of domain 
theory~\cite{GaboardiEtAl/POPL/2017,DBLP:journals/tcs/BaierM94,DeBakker/Semantics-concurrency/1982,Escardo/Metric-model-PCF/1999,Arnold/Metric-interpretations/1980} 
(see also \cite{Breugel/Introduction-metric-semantics/2001} 
for an introduction to the subject)
and coinduction~\cite{DBLP:conf/icalp/BreugelW01,Breugel-Worrell/Behavioural-pseudometric-probabilistic-transition-systems/2005,DBLP:journals/tcs/DesharnaisGJP04,DBLP:conf/concur/ChatzikokolakisGPX14}. The common denominator among all these
approaches is that on one hand, the notion of a congruence,
crucial for compositional reasoning, is replaced by the one of a \emph{Lipschitz-continuous}
map: any context should not amplify (too much) the distance between 
any pair of terms, when it is fed with either the former or the latter:
$$
\delta(C[M],C[N])\leq c\cdot\delta(M,N).
$$
This enforces compositionality, and naturally leads us to consider
metric spaces and Lipschitz functions as the underlying category.
As is well known, this is not a cartesian closed category, and
thus does \emph{not} form a model of typed $\lambda$-calculi, unless one
adopts linear type systems, or type systems in which the number
of uses of each variable is kept track of, like \FUZZ~\cite{Pierce/DistanceMakesTypesGrowStronger/2010}. This
somehow limits the compositionality of the metric 
approach~\cite{GaboardiEtAl/POPL/2017,Gavazzo/LICS/2018}.

\newcommand{\eqdef}{:=}
There are however program transformations which are intrinsically 
unjustifiable in the metric approach. Consider the following two
programs of type $\reals\arr\reals$
$$
\lsin\eqdef\lambda x.\mathtt{sin}(x)\qquad\qquad \lid\eqdef\lambda x.x.
$$ 
The two terms compute two very different functions on the real
numbers, namely the sine trigonometric function and the
identity on $\rset$, respectively. The distance
$\mid\sin x-x\mid$ is unbounded when $x$ ranges over
$\rset$. As a consequence, the numerical distance between $\lsin$ and $\lid$,
however defined, is infinite, and the program transformation turning
$\lsin$ into $\lid$ cannot be justified this way, for very good reasons. As
highlighted by Westbrook and Chaudhuri~\cite{WestbrookAndChaudhuri}, 
this is not the end of
the story, at least if the environment in which $\lsin$ and $\lid$ operate
feed either of them \emph{only with} real numbers close to $0$,
then $\lsin$ can be substituted with $\lid$ without affecting \emph{too
much} the overall behaviour of the system. 

The key insight by Westbrook and Chaudhuri is that justifying program
transformations like the one above requires taking the difference
$\delta(\lsin,\lid)$ between $\lsin$ and $\lid$ not merely as a number, but as
a more structured object. What they suggest is to take $\delta(\lsin,\lid)$
as \emph{yet another program}, which however describes the difference between $\lsin$ and
$\lid$:
$$
\delta(\lsin,\lid)\eqdef\lambda x.\lambda \varepsilon.|\sin x-x|+\varepsilon.
$$
This reflects the fact that the distance between $\lsin$ and $\lid$, namely
the discrepancy between their output, depends not only on the
discrepancy on the input, namely on $\varepsilon$, but also \emph{on the input
itself}, namely on $x$. It both $x$ and $\varepsilon$ are close to $0$,
$\delta(\lsin,\lid)$ is itself close to $0$.

In this paper, we develop Westbrook and Chaudhuri's ideas, and
turn them into a framework of \emph{differential logical
  relations}. We will do all this in a simply-typed $\lambda$-calculus
with real numbers as the only base type. Starting from such a minimal
calculus has at least two advantages: on the one hand one can talk
about meaningful examples like the one above, and on the other hand the
induced metatheory is simple enough to highlight the key concepts.

The contributions of this paper can be summarised as follows:
\begin{varitemize}
\item
  After introducing our calculus $\STlamreal$, we define differential
  logical relations inductively on types, as ternary relations between
  pairs of programs and \emph{differences}. The latter are mere set theoretic
  entities here, and the nature of differences between terms depends
  on terms' types.
\item
  We prove a soundness theorem for differential logical relations,
  which allows us to justify compositional reasoning about terms'
  differences. We also prove a \emph{finite difference theorem}, which
  stipulates that the distance between two simply-typed
  $\lambda$-terms is finite if mild conditions hold on the underlying
  set of function symbols.
\item
  We give embeddings of logical and metric relations into differential
  logical relations. This witnesses that the latter are a
  generalisation of the former two.
\item
  Finally, we show that generalised metric domains, the mathematical
  structure underlying differential logical relations, form a cartesian
  closed category, contrarily to the category of metric spaces, which is
  well known not to have the same property.
\end{varitemize}
Due to space constraints, many details have to be omitted, but can be
found in an Extended Version of this work~\cite{EV}.
\section{A Simply-Typed $\lambda$-Calculus with Real Numbers}
In this section, we introduce a simply-typed $\lambda$-calculus in which
the only base type is the one of real numbers, and constructs for iteration
and conditional are natively available. The choice of this language as
the reference calculus in this paper has been made for the sake of
simplicity, allowing us to concentrate on the most crucial aspects,
at the same time guaranteeing a minimal expressive power.

\paragraph*{Terms and Types}
$\STlamreal$ is a typed $\lambda$-calculus, so its definition
starts by giving the language of \emph{types}, which is defined as follows:
$$
\tone,\ttwo \bnf \reals
\mathbin{\bigl\vert} \tone\arr\ttwo
\mathbin{\bigl\vert} \tone\prd\ttwo.
$$
The expression $\tone^n$ stands for
$\underbrace{\tone\prd\cdots\prd\tone}_{\mbox{$n$ times}}$.
The set of \emph{terms} is defined as follows:
\begin{align*}
\lone,\ltwo &\bnf\vone
\mathbin{\bigl\vert} \rone
\mathbin{\bigl\vert} \pfone{n}
\mathbin{\bigl\vert} \abs{\vone}{\lone}
\mathbin{\bigl\vert} \lone\ltwo
\mathbin{\bigl\vert} \pair{\lone}{\ltwo}
\mathbin{\bigl\vert} \prjleft\midd\prjright
\mathbin{\bigl\vert} \iflz{\lone}{\ltwo}
\mathbin{\bigl\vert} \iter{\lone}{\ltwo}
\end{align*}
where $\vone$ ranges over a set $\vset$ of variables, $\rone$ ranges
over the set $\rset$ of real numbers, $n$ is a natural number
and $\pfone{n}$ ranges over a set $\rfuncs{n}$
of total real functions of arity $n$. We do not make any
assumption on $\{\rfuncs{n}\}_{n\in\nset}$, apart from
the predecessor $\prede$ being part of $\rfuncs{1}$. 
The family, in particular, could in principle contain non-continuous functions.
The expression $\tuple{\lone_1,\ldots,\lone_n}$ is simply
a shortcut for
$\pair{\ldots\pair{\pair{\lone_1}{\lone_2}}{\lone_3}\ldots}{\lone_n}$.
All constructs are self-explanatory, except for the
$\mathtt{ifz}$ and $\mathtt{iter}$ operators, which are
conditional and iterator combinators, respectively.
An \emph{environment} $\eone$ is a set of assignments
of types to variables in $\vset$ where each variable occurs
at most once. A \emph{type judgment} has the form
$\tj{\eone}{\lone}{\tone}$ where $\eone$ is an environment,
$\lone$ is a term, and $\tone$ is a type. Rules for deriving correct
typing judgments are in Figure~\ref{fig:typingrules}, and are standard.
\begin{figure}
\begin{center}
  \fbox{
  \footnotesize
  \begin{minipage}{.97\textwidth}
    $$
    \infer
        {\tj{\eone}{\vone}{\tone}}
        {\vone:\tone\in\eone}
    \qquad\quad
    \infer
        {\tj{\eone}{\rone}{\reals}}
        {}
    \qquad\quad
    \infer
        {\tj{\eone}{\pfone{n}}{\reals^n\arr\reals}}
        {\pfone{n} \in \rfuncs{n}}
    \qquad\quad
    \infer
        {\tj{\eone}{\abs{\vone}{\lone}}{\tone\arr\ttwo}}
        {\tj{\eone,\vone:\tone}{\lone}{\ttwo}}
    $$
    $$    
  \infer
      {\tj{\eone}{\lone\ltwo}{\ttwo}}
      {\tj{\eone}{\lone}{\tone\arr\ttwo} & \tj{\eone}{\ltwo}{\tone}}
    \qquad
  \infer
      {\tj{\eone}{\pair{\lone}{\ltwo}}{\tone\prd\ttwo}}
      {\tj{\eone}{\lone}{\tone} & \tj{\eone}{\ltwo}{\ttwo}}
  \qquad
  \infer
      {\tj{\eone}{\prjleft}{\tone\prd\ttwo\arr\tone}}
      {}      
  \qquad
  \infer
      {\tj{\eone}{\prjright}{\tone\prd\ttwo\arr\ttwo}}
      {}
  $$
  $$
  \infer
      {\tj{\eone}{\iflz{\lone}{\ltwo}}{\reals\arr\tone}}
      {\tj{\eone}{\lone}{\tone} & \tj{\eone}{\ltwo}{\tone}}
  \qquad\qquad\qquad
  \infer
      {\tj{\eone}{\iter{\lone}{\ltwo}}{\reals\arr\tone}}
      {\tj{\eone}{\lone}{\tone\arr\tone} & \tj{\eone}{\ltwo}{\tone}}      
  $$    
  \end{minipage}}
  \caption{Typing rules for $\STlamreal$.}\label{fig:typingrules}
\end{center}
\end{figure}
The set of terms $\lone$ for which $\tj{\cdot}{\lone}{\tone}$ is
derivable is indicated as $\ct{\tone}$.
\paragraph*{Call-by-Value Operational Semantics}
A static semantics is of course not enough to give meaning to a
paradigmatic programming language, the dynamic aspects being captured
only once an \emph{operational} semantics is defined. The latter
turns out to be very natural. \emph{Values} are defined as follows:
$$
\vlone,\vltwo \bnf \rone 
\mathbin{\bigl\vert} \pfone{n}
\mathbin{\bigl\vert} \abs{\vone}{\lone}
\mathbin{\bigl\vert} \pair{\lone}{\ltwo}
\mathbin{\bigl\vert} \prjleft
\mathbin{\bigl\vert} \prjright
\mathbin{\bigl\vert} \iflz{\lone}{\ltwo}
\mathbin{\bigl\vert} \iter{\lone}{\ltwo}
$$
The set of closed values of type $\tone$ is $\cv{\tone}\subseteq\ct{\tone}$,
and the evaluation of $\lone\in\ct{\tone}$ produces a value
$\vlone\in\cv{\tone}$, as formalised by the rules in Figure~\ref{fig:evaluation},
through the judgment $\eval{\lone}{\vlone}$.
\begin{figure}
\begin{center}
  \fbox{
  \footnotesize
  \begin{minipage}{.97\textwidth}
  $$
  \infer
      {\eval{\vlone}{\vlone}}
      {}
  \qquad
  \infer
      {\eval{\lone\ltwo}{\fone(\rone_1,\ldots,\rone_n)}}
      {\eval{\lone}{\pfone{n}} & \eval{\ltwo}{\tpl{\lthree}{n}} & \eval{\lthree_i}{\rone_i}}
  \qquad
  \infer
      {\eval{\lone\ltwo}{\vltwo}}
      {\eval{\lone}{\abs{\vone}{\lthree}} & \eval{\ltwo}{\vlone} & \eval{\sbs{\lthree}{\vlone}{\vone}}{\vltwo}}
  $$
  $$
  \infer
      {\eval{\lone\ltwo}{\vlone}}
      {\eval{\lone}{\prjleft} & \eval{\ltwo}{\pair{\lthree}{\lfour}} & \eval{\lthree}{\vlone}}
  \qquad
  \infer
      {\eval{\lone\ltwo}{\vlone}}
      {\eval{\lone}{\prjright} & \eval{\ltwo}{\pair{\lthree}{\lfour}} & \eval{\lfour}{\vlone}}      
  $$
  $$    
  \infer
      {\eval{\lone\ltwo}{\vlone}}
      {\eval{\lone}{\iflz{\lthree}{\lfour}} & \eval{\ltwo}{\rone} & r<0 & \eval{\lthree}{\vlone}}
  \qquad
  \infer
      {\eval{\lone\ltwo}{\vlone}}
      {\eval{\lone}{\iflz{\lthree}{\lfour}} & \eval{\ltwo}{\rone} & r\geq 0 & \eval{\lfour}{\vlone}}
  $$
  $$    
  \infer
      {\eval{\lone\ltwo}{\vlone}}
      {\eval{\lone}{\iter{\lthree}{\lfour}} & \eval{\ltwo}{\rone} & r<0 & \eval{\lfour}{\vlone}}
  $$
  $$    
  \infer
      {\eval{\lone\ltwo}{\vlone}}
      {\eval{\lone}{\iter{\lthree}{\lfour}} & \eval{\ltwo}{\rone} & r\geq 0 & \eval{\lthree((\iter{\lthree}{\lfour})(\prede(r))}{\vlone}}
  $$  
  \end{minipage}}
\end{center}
\caption{Operational semantics for $\STlamreal$.} \label{fig:evaluation}
\end{figure}
We write $\terminate{\lone}$ if $\eval{\lone}{\vlone}$ is derivable \emph{for some} $\vlone$.
The absence of full recursion has the nice consequence of guaranteeing a form of termination:
\begin{theorem}\label{thm:norm}
  The calculus $\STlamreal$ is terminating:
  if ${\tj{\cdot}{\lone}{\tone}}$ then $\terminate{\lone}$.
\end{theorem}
We show the normalisation theorem using the standard
reducibility candidate argument.
\begin{definition}
  We define $\redset{\tone}$ as follows.
  \begin{align*}
    \redset{\tone}&= \{\lone \;\mid\; \eval{\lone}{\vlone}\land \vlone\in\vredset{\tone}\}\\
    \vredset{\reals}&= \rset\\
    \vredset{\tone\arr\ttwo}&= \{\vlone \;\mid\; \forall\vltwo\in\vredset{\tone}. \vlone\vltwo\in\redset{\ttwo}\}\\
    \vredset{\tone\prd\ttwo}&= \{\pair{\lone}{\ltwo} \;\mid\; \lone\in\redset{\tone} \land \ltwo\in\redset{\ttwo}\}
  \end{align*}
\end{definition}
Then we show the two lemmas that prove Theorem~\ref{thm:norm} together:
\begin{lemma}
\label{lemma:normalization}
  If ${\tj{\cdot}{\lone}{\tone}}$, then $\lone \in \redset{\tone}$.
\end{lemma}
\begin{proof}
  The following strengthening of the statement can be proved by
  incuction on the structure of $\lone$: whenever
  ${\tj{\vone_1:\ttwo_1,\ldots,\vone_n:\ttwo_n}{\lone}{\tone}}$
  and whenever $\vlone_i\in\vredset{\ttwo_i}$ it holds that
  $$
  \lone\{\vlone_1/\vone_1,\ldots,\vlone_n/\vone_n\}\in\redset{\tone}.
  $$
  All inductive cases are standard. One that deserves a little bit of
  attentionis the case of abstractions, in which one needs to deduce
  that $(\abs{\vone}{\lone})\vlone\in\redset{\tone}$ iff
  $\sbs{\lone}{\vone}{\vlone}\in\redset{\tone}$. But this is of course
  easy to prove, because the former evaluated to any value $\vltwo$
  iff the latter evaluates to $\vltwo$. Another delicate case is
  the one of $\iter{\lone}{\ltwo}$, which requires an induction.
\end{proof}
Please observe that, by definition, if $\lone\in\redset{\tone}$, then
$\terminate{\lone}$. As a consequence, one easily gets termination
from Lemma~\ref{lemma:normalization}.
\begin{corollary}\label{cor:realnormal}
  If ${\tj{\cdot}{\lone}{\reals}}$ then there exists a unique $\rone \in \rset$
  satisfying $\eval{\lone}{\rone}$, which we indicate as $\nf{\lone}$.
\end{corollary}
\begin{proof}
  By Theorem~\ref{thm:norm} there exists a value $\vlone$ satisfying
  $\eval{\lone}{\vlone}$.  The only form of value of type $\reals$ is
  $\rone \in \rset$. Moreover, the fact that such a $\rone$ is unique
  is a consequence of the following, slightly more general result: if
  $\eval{\lone}{\vlone}$ and $\eval{\lone}{\vltwo}$, then $\vlone$ is
  syntactically equal to $\vltwo$. This can be proved by a
  straightforward induction on the structure of the proof that, e.g.,
  $\eval{\lone}{\vlone}$.
\end{proof}
\paragraph*{Context Equivalence}
A \emph{context} $\cone$ is nothing more than a term containing a
single occurrence of a placeholder $\chole$. Given a context $\cone$,
$\cone[\lone]$ indicates the term one obtains by substituting $\lone$
for the occurrence of $\chole$ in $\cone$. Typing rules in
Figure~\ref{fig:typingrules} can be lifted to contexts by a generalising
judgments to the form $\tj{\eone}{\cone[\tj{\etwo}{\cdot}{\tone}]}{\ttwo}$,
by which one captures that whenever $\tj{\etwo}{\lone}{\tone}$
it holds that $\tj{\eone}{\cone[\lone]}{\ttwo}$. Two terms $\lone$
and $\ltwo$ such that $\tj{\eone}{\lone,\ltwo}{\tone}$ are said to
be \emph{context equivalent} \cite{Morris/PhDThesis} 
when for every $\cone$ such that
$\tj{\emenv}{\cone[\tj{\eone}{\cdot}{\tone}]}{\reals}$ it holds
that $\nf{\cone[\lone]}=\nf{\cone[\ltwo]}$.
Context equivalence is the largest adequate congruence, and is thus
considered as the coarsest ``reasonable'' equivalence between terms.
It can also be turned into a 
pseudometric~\cite{DBLP:conf/esop/CrubilleL17,DBLP:conf/lics/CrubilleL15} 
--- called \emph{context distance} ---
by stipulating that
$$
\delta(\lone,\ltwo)=\sup_{\tj{\emenv}{\cone[\tj{\eone}{\cdot}{\tone}]}{\reals}}
  |\nf{\cone[\lone]}-\nf{\cone[\lone]}|.
$$
  The obtained notion of distance, however, is bound to 
  trivialise~\cite{DBLP:conf/esop/CrubilleL17}, given
that $\STlamreal$ is not affine. Trivialisation of context distance
highlights an important limit of the metric approach to program 
difference which, ultimately, can be identified with the fact that 
program distances are sensitive to interactions with the environment. 
Our notion of a differential logical relation tackles such a problem 
from a different perspective, namely refining the concept of a 
program difference which is not just a number, but is now able to 
take into account interactions with the environment.
\paragraph*{Set-Theoretic Semantics}
Before introducing differential logical relations, it is useful to 
remark that we can give $\STlamreal$ a standard set-theoretic semantics. 
To any type $\tone$ we associate the set $\sem{\tone}$, the latter being 
defined by induction on the structure of $\tone$ as follows:
$$
\sem{\reals}=\rset;\qquad\qquad
\sem{\tone\arr\ttwo}=\sem{\tone}\arr\sem{\ttwo};\qquad\qquad
\sem{\tone\prd\ttwo}=\sem{\tone}\prd\sem{\ttwo}.
$$
This way, any closed term $\lone\in\ct{\tone}$ is interpreted as an
element $\sem{\lone}$ of $\sem{\tone}$ in a natural way (see, e.g.~\cite{Mitchell1996}).
Up to now, everything we have said about~$\STlamreal$ is absolutely
standard, and only serves to set the stage for the next sections.
\section{Making Logical Relations Differential}
Logical relations can be seen as one of the \emph{many} ways of
defining when two programs are to be considered equivalent. Their
definition is type driven, i.e., they can be seen as a \emph{family}
$\{\relone_\tone\}_{\tone}$ of binary relations indexed by types such
that $\relone_\tone\subseteq\ct{\tone}\times\ct{\tone}$. This section
is devoted to showing how all this can be made into differential logical
relations.

The first thing that needs to be discussed is how to define the space
of \emph{differences} between programs. These are just boolean values
in logical relations, become real numbers in ordinary metrics, and is
type-dependent itself here.
A function $\metdom{\cdot}$ that assigns a set to each type is defined
as follows:
$$
\metdom{\reals}= \pirset;\qquad\qquad
\metdom{\tone\arr\ttwo}=\sem{\tone}\prd\metdom{\tone}\arr\metdom{\ttwo};\qquad\qquad
\metdom{\tone\prd\ttwo}=\metdom{\tone}\prd\metdom{\ttwo};
$$
where $\rset_{\geq 0}^{\infty} = \rset_{\geq 0}\cup\{\infty\}$.
The set $\metdom{\tone}$ is said to be the \emph{difference space}
for the type $\tone$ and is meant to model the outcome of comparisons
between closed programs of type $\tone$. As an example, when
$\tone$ is $\reals\arr\reals$, we have that
$\metdom{\tone}=\rset\times\pirset\arr\pirset$. This is the type of
the function $\delta(\lone,\ltwo)$ we used to compare the two
programs described in the Introduction.

Now, which structure could we endow $\metdom{\tone}$ with? First of
all, we can define a partial order
${\leq_{\tone}}$ over $\metdom{\tone}$ for each type $\tone$ as follows:
\begin{align*}
  \rone&\leq_{\reals}\rtwo
  && \text{ if } \rone\leq\rtwo
  \text{ as the usual order over $\rset_{\geq 0}^{\infty}$};
  \\
  \fone&\leq_{\tone\arr\ttwo}\ftwo
  && \text{ if } \forall\elone\in\sem{\tone}. \forall\elemone\in\metdom{\tone}.
  \fone(\elone,\elemone) \leq_{\ttwo} \ftwo(\elone,\elemone);
  \\
  (\elemone,\elemtwo)&\leq_{\tone\prd\ttwo}(\elemthree,\elemfour)
  && \text{ if } \elemone\leq_{\tone}\elemthree \text{ and } \elemtwo\leq_{\ttwo}\elemfour.
\end{align*}
This order has least upper bounds and greater lower bounds, thanks to
the nice structure of $\pirset$:
\begin{proposition}\label{prop:complete}
  For each type $\tone$, $(\metdom{\tone}, \leq_{\tone})$ forms a complete
  lattice.
\end{proposition}
  \begin{proof}
    We show that each $\metdom{\tone}$ has suprema
    by induction on types.
    \begin{varitemize}
    \item Case $\reals$.
      Then $(\metdom{\tone}, \leq_{\tone}) = (\rset_{\geq 0}^{\infty}, \leq)$
      is clearly complete.
    \item Case $\tone\arr\ttwo$.
      Given a subset $\setone \subseteq \metdom{\tone\arr\ttwo}$, we define
      $\elemthree_{\setone} \in \metdom{\tone\arr\ttwo}$ as:
      $$
      \elemthree_{\setone}(\vlone,\elemone)
      = \sup_{\fone \in \setone}\fone(\vlone,\elemone),
      $$
      where the supremum on the right-hand side exists by induction hypothesis (on the type $\ttwo$).
      This $\elemthree_{\setone}$ serves as the supremum of $\setone$ because:
      \begin{varitemize}
      \item (\textbf{Upperbound}.) For any $\ftwo \in \setone$, 
        by definition of supremum it holds that:
        
        $\forall\vlone\in\cv{\tone}. \forall\elemone\in\metdom{\tone}.
        \ftwo(\vlone,\elemone)
        \leq_{\tone} \sup_{\fone \in \setone}\fone(\vlone,\elemone)
        = \elemthree_{\setone}(\vlone,\elemone)$.
        Hence $\ftwo \leq_{\tone\arr\ttwo}\elemthree_{\setone}$.
      \item (\textbf{Leastness}.) Suppose that $\elemthree'$
        is an upperbound of $\setone$,
        i.e.\ $\forall\fone\in\setone. \fone \leq_{\tone\arr\ttwo} \elemthree'$.
        Then, it by definition means that:
        $\forall\fone\in\setone. \forall\vlone\in\cv{\tone}. \forall\elemone\in\metdom{\tone}.
        \fone(\vlone,\elemone) \leq_{\ttwo} \elemthree'(\vlone,\elemone).$
        Therefore $\elemthree'(\vlone,\elemone)$ is an upperbound of
        the set $\{\fone(\vlone,\elemone)\}_{\fone\in\setone}$
        for each $\vlone, \elemone$.
            Thus by definition of supremum,
            $\elemthree_{\setone}(\vlone,\elemone)
            = \sup_{\fone \in \setone}\fone(\vlone,\elemone)
            \leq_{\ttwo} \elemthree'(\vlone,\elemone)$
            for each $\vlone, \elemone$.
            Hence $\elemthree_{\setone} \leq_{\tone\arr\ttwo} \elemthree'$
            holds by definition of $\leq_{\tone\arr\ttwo}$.
        \end{varitemize}
        \item Case $\tone \prd \ttwo$.
          Given a subset $\setone \subseteq \metdom{\tone\prd\ttwo}$,
          we define 
          $
          \elemthree_{\setone}
          = (\sup \boldsymbol{\prjleft}\setone,
             \sup \boldsymbol{\prjright}\setone)
          \in \metdom{\tone\prd\ttwo}$,
        where $\boldsymbol{\prjleft}$ and $\boldsymbol{\prjright}$ are
        meta-level projections and the suprema on the right-hand side
        exist by induction hypothesis (on the types $\tone$ and $\ttwo$).
        One can verify that $\elemthree_{\setone}$ is the supremum of
        $\setone$ in a straightforward way:
        \begin{varitemize}
          \item (\textbf{Upperbound}.)
            For any $(\elemone,\elemtwo) \in \setone$, 
            by definition of supremum 
            $\elemone \leq_{\tone} \sup \boldsymbol{\prjleft}\setone$
            and
            $\elemtwo \leq_{\ttwo} \sup \boldsymbol{\prjright}\setone$
            hold.
            Hence $(\elemone,\elemtwo) \leq_{\tone\prd\ttwo}
            (\sup \boldsymbol{\prjleft}\setone, \sup \boldsymbol{\prjright}\setone)
            = \elemthree_{\setone}$ by definition of $\leq_{\tone\prd\ttwo}$.
          \item (\textbf{Leastness}.)
            Suppose that $(\elemthree'_1, \elemthree'_2)$
            is an upperbound of $\setone$,
            i.e.\

            \noindent
            $\forall(\elemone_1,\elemone_2)\in\setone.
            (\elemone_1,\elemone_2) \leq_{\tone\prd\ttwo} (\elemthree'_1, \elemthree'_2)$.
            It by definition means that
            $\forall \elemone_1\in\boldsymbol{\prjleft}\setone.
            \elemone_1 \leq_{\tone} \elemthree'_1$
            and $\forall \elemone_2\in\boldsymbol{\prjright}\setone.
            \elemone_2 \leq_{\ttwo} \elemthree'_2$.
            Therefore $\elemthree'_1$ (resp.\ $\elemthree'_2$) is
            an upperbound of the set $\boldsymbol{\prjleft}\setone$
            (resp.\ $\boldsymbol{\prjright}\setone$).
            Thus by definition of supremum,
            $\sup \boldsymbol{\prjleft}\setone\leq_{\tone}\elemthree'_1$
            (resp.\ $\sup \boldsymbol{\prjright}\setone\leq_{\ttwo}\elemthree'_2$).
            Hence $(\sup \boldsymbol{\prjleft}\setone, \sup \boldsymbol{\prjright}\setone) \leq_{\tone\prd\ttwo} (\elemthree'_1, \elemthree'_2)$
            holds by definition of $\leq_{\tone\prd\ttwo}$.
        \end{varitemize}
    \end{varitemize}
  \end{proof}

The fact that $\metdom{\tone}$ has a nice order-theoretic structure
is not the end of the story. 
For every type $\tone$, we define a binary operation
  $\mult{\tone}$ as follows:
  \begin{align*}
    \rone \mult{\reals} \rtwo &\stackrel{\text{def.}}{=} \rone+\rtwo\mbox{ if }\rone, \rtwo \in \rset_{\geq 0};
    &&&
    (\fone \mult{\tone\arr\ttwo} \ftwo)(\vlone,\elemone)
    &\stackrel{\text{def.}}{=} \fone(\vlone,\elemone) \mult{\ttwo} \ftwo(\vlone,\elemone);
    \\
    \rone \mult{\reals} \rtwo &\stackrel{\text{def.}}{=} \infty\mbox{ if }\rone = \infty\vee\rtwo = \infty;
    &&&
     (\elemone,\elemthree) \mult{\tone\prd\ttwo} (\elemtwo,\elemfour)
    &\stackrel{\text{def.}}{=} 
    (\elemone\mult{\tone}\elemtwo, \elemthree\mult{\ttwo}\elemfour).
  \end{align*}
  This is precisely what it is needed to turn $\metdom{\tone}$ into a 
  \emph{quantale}\footnote{
  Recall that a quantale $\qtlone = (Q, \leq_Q, 0_Q, *_Q)$ consists 
  of a complete lattice $(Q, \leq_Q)$ and a monoid $(Q, 0_Q, *_Q)$ 
  such that the lattice and monoid structure properly interact 
  (meaning that monoid multiplication distributes over joins). 
  We refer to 
  \cite{Rosenthal/Quantales/1990,Hoffman-Seal-Tholem/monoidal-topology/2014} 
  for details.
  }
  \cite{Rosenthal/Quantales/1990}.
  \begin{proposition}\label{prop:quantale}
    For each type $\tone$, $\metdom{\tone}$ forms
    a commutative unital non-idempotent quantale.
    That is, the following holds for any $\tone$:
    \begin{varitemize}
      \item $\elemone \mult{\tone} \elemtwo = \elemtwo \mult{\tone} \elemone$
        for all $\elemone, \elemtwo \in \metdom{\tone}$,
      \item $\elemone \mult{\tone} (\sup_{\indone \in \indset}\elemtwo_{\indone})
        = \sup_{\indone \in \indset} (\elemone \mult{\tone}\elemtwo_{\indone})$
        for all $\elemone, \elemtwo_{\indone} \in \metdom{\tone}$
        where $\indset$ is an arbitrary index set,
      \item there exists an element $\unit{\tone} \in \metdom{\tone}$
        satisfying $\unit{\tone} \mult{\tone} \elemone = \elemone$
        for all $\elemone \in \metdom{\tone}$,
      \item $\mult{\tone}$ does not necessarily satisfy
        $\elemone \mult{\tone} \elemone = \elemone$.
    \end{varitemize}
  \end{proposition}
  \begin{proof}
    By induction on $\tone$.
    \begin{varitemize}
    \item
      Case $\reals$.
        The multiplication $\mult{\reals}$ is clearly commutative and satisfies
        $\rone \mult{\tone} (\sup_{\indone \in \indset}\rtwo_{\indone})
        = \sup_{\indone \in \indset} (\rone \mult{\tone}\rtwo_{\indone})$
        for all $\rone, \rtwo_{\indone} \in \metdom{\reals} = \rset_{\geq 0}^{\infty}$.
        The unit $\unit{\reals}$ is $0$;
        the multiplication is obviously non-idempotent.
      \item
        Case $\tone \arr \ttwo$.
        It holds that $(\fone \mult{\tone\arr\ttwo} \ftwo) = (\ftwo \mult{\tone\arr\ttwo} \fone)$ because
        $$
        \begin{array}{ll}  
          (\fone \mult{\tone\arr\ttwo} \ftwo)(\vlone,\elemone)
          &= \fone(\vlone,\elemone) \mult{\ttwo} \ftwo(\vlone,\elemone)\\
          &\stackrel{\text{I.H.}}{=} \ftwo(\vlone,\elemone) \mult{\ttwo} \fone(\vlone,\elemone)\\
          &=(\ftwo \mult{\tone\arr\ttwo} \fone)(\vlone,\elemone)
        \end{array}
        $$

        and $\fone \mult{\tone\arr\ttwo} (\sup_{\indone \in \indset}(\ftwo_{\indone})) = \sup_{\indone \in \indset} (\fone \mult{\tone\arr\ttwo} \ftwo_{\indone})$
        because

        $$
        \begin{array}{ll} 
          (\fone \mult{\tone\arr\ttwo} (\sup_{\indone \in \indset}(\ftwo_{\indone})))(\vlone, \elemone)
          &=\fone(\vlone,\elemone) \mult{\ttwo} (\sup_{\indone \in \indset}\ftwo_{\indone})(\vlone,\elemone)\\
          &=\fone(\vlone,\elemone) \mult{\ttwo} (\sup_{\indone \in \indset}(\ftwo_{\indone}(\vlone,\elemone)))\\
          &\stackrel{\text{I.H.}}{=}\sup_{\indone \in \indset}(\fone(\vlone,\elemone) \mult{\ttwo}\ftwo_{\indone}(\vlone,\elemone))\\
          &=\sup_{\indone \in \indset}((\fone \mult{\tone\arr\ttwo}\ftwo_{\indone})(\vlone,\elemone))\\
          &=(\sup_{\indone \in \indset}(\fone \mult{\tone\arr\ttwo}\ftwo_{\indone}))(\vlone,\elemone).
        \end{array}
        $$
        The unit $\unit{\tone\arr\ttwo}$ is the constant $\unit{\ttwo}$ function: $\unit{\tone\arr\ttwo}(\vlone,\elemone) = \unit{\ttwo}$ for all $\vlone, \elemone$.

      \item Case $\tone \prd \ttwo$.
        Then 
        $$
        \begin{array}{ll} 
          (\elemone,\elemone') \mult{\tone\prd\ttwo} (\elemtwo,\elemtwo')
          &= (\elemone\mult{\tone}\elemtwo, \elemone'\mult{\ttwo}\elemtwo')\\
          &\stackrel{\text{I.H.}}{=} (\elemtwo\mult{\tone}\elemone, \elemtwo'\mult{\ttwo}\elemone')\\
          &=(\elemtwo,\elemtwo') \mult{\tone\prd\ttwo} (\elemone,\elemone')
        \end{array}
        $$
        and 
        $$
        \begin{array}{ll} 
          (\elemone,\elemone') \mult{\tone\prd\ttwo} (\sup_{\indone \in \indset}(\elemtwo_{\indone}, \elemtwo'_{\indone}))
          &=(\elemone,\elemone') \mult{\tone\prd\ttwo} (\sup_{\indone \in \indset}\elemtwo_{\indone}, \sup_{\indone \in \indset}\elemtwo'_{\indone})\\
          &=(\elemone \mult{\tone} (\sup_{\indone \in \indset}\elemtwo_{\indone}) ,\elemone' \mult{\tone} (\sup_{\indone \in \indset}\elemtwo'_{\indone}))\\
          &\stackrel{\text{I.H.}}{=}(\sup_{\indone \in \indset}(\elemone \mult{\tone}\elemtwo_{\indone}), \sup_{\indone \in \indset}(\elemone' \mult{\tone}\elemtwo'_{\indone}))\\
          &=\sup_{\indone \in \indset}(\elemone \mult{\tone}\elemtwo_{\indone}, \elemone' \mult{\tone}\elemtwo'_{\indone})\\
          &=\sup_{\indone \in \indset}\left((\elemone,\elemone') \mult{\tone\prd\ttwo} (\elemtwo_{\indone}, \elemtwo'_{\indone})\right).
        \end{array}
        $$
        The unit $\unit{\tone\prd\ttwo}$ is $(\unit{\tone}, \unit{\ttwo})$.
    \end{varitemize}    
  \end{proof}

  The fact that $\metdom{\tone}$ is a quantale means that it has, e.g., the right
  structure to be the codomain of generalised 
  metrics~\cite{Lawvere/GeneralizedMetricSpaces/1973,Hoffman-Seal-Tholem/monoidal-topology/2014}. Actually,
  a more general structure is needed for our purposes, namely the one of a generalised
  metric domain, which will be thoroughly discussed in Section~\ref{sect:categorical} below.
  For the moment, let us concentrate our attention to programs:

  \begin{definition}[Differential Logical Relations]\label{def:dlr}
  We define a \emph{differential logical relation}
  $\{\delta_{\tone} \subseteq \lset{\tone}\times\metdom{\tone}\times\lset{\tone}\}_{\tone}$
  as a set of ternary relations indexed by types satisfying
  \begin{align*}
    \diffmet{\reals}{\lone}{\rone}{\ltwo}
     \Leftrightarrow &\; |\nf{\lone}-\nf{\ltwo}|\leq\rone;\\
    \diffmet{\tone\prd\ttwo}{\lone}{(\done_1, \done_2)}{\ltwo}
     \Leftrightarrow &\; \diffmet{\tone}{\prjleft\lone}{\done_1}{\prjleft\ltwo}
    \land \diffmet{\ttwo}{\prjright\lone}{\done_2}{\prjright\ltwo}\\
    \diffmet{\tone\arr\ttwo}{\lone}{\done}{\ltwo}
     \Leftrightarrow &\; (\forall \vlone \in \cv{\tone}.\ \forall \vone \in \metdom{\tone}.\ \forall \vltwo \in \cv{\tone}.\ \\
    & \,\,\diffmet{\tone}{\vlone}{\vone}{\vltwo}
    \Rightarrow \diffmet{\ttwo}{\lone\vlone}{\done(\sem{\vlone},\vone)}{\ltwo\vltwo}
    \land \diffmet{\ttwo}{\lone\vltwo}{\done(\sem{\vlone},\vone)}{\ltwo\vlone}).
  \end{align*}
  \end{definition}
  An intuition behind the condition required for
  $\diffmet{\tone\arr\ttwo}{\lone}{\done}{\ltwo}$ is that
  $\done(\sem{\vlone},\vtwo)$ overapproximates both the ``distance''
  between $\lone \vlone$ and $\ltwo \vltwo$ and the one between $\lone
  \vltwo$ and $\ltwo \vlone$, this \emph{whenever} $\vltwo$ is within
  the error $\vone$ from $\vlone$.

    Some basic facts about differential logical relations, which will be useful in the following
    are now in order.
  \begin{lemma}\label{lem:symmetric}
    For every $\tone,\lone,\ltwo$, it holds that
    $\diffmet{\tone}{\lone}{\done}{\ltwo}$ if and only if
    $\diffmet{\tone}{\ltwo}{\done}{\lone}$.
  \end{lemma}
  \begin{proof}
    By induction on types.
  \end{proof}
  \begin{lemma}
    Let $\tj{\cdot}{\lone,\ltwo}{\tone}$ and $\done \in \metdom{\tone}$.
    Assume that 
    $\diffmet{\reals}{\lone}{\done}{\ltwo}$
    and $\diffmet{\reals}{\ltwo}{\dtwo}{\lthree}$
    implies $\diffmet{\reals}{\lone}{\done+\dtwo}{\lthree}$.
    Then
    $\diffmet{\tone}{\lone}{\done}{\ltwo}$
    and $\diffmet{\tone}{\ltwo}{\dtwo}{\lthree}$
    implies $\diffmet{\tone}{\lone}{\done+\dtwo}{\lthree}$.
  \end{lemma}
  \begin{proof}
    By induction on types.
  \end{proof}
  \begin{lemma}
    Let $\tj{\cdot}{\lone,\ltwo}{\tone}$ and $\done \in \metdom{\tone}$.
    Assume that 
    $\diffmet{\reals}{\lone}{\done}{\ltwo}$
    and $\done \leq \dtwo$
    implies $\diffmet{\reals}{\lone}{\dtwo}{\ltwo}$.
    Then
    $\diffmet{\tone}{\lone}{\done}{\ltwo}$
  and $\done \leq \dtwo$
  implies $\diffmet{\tone}{\lone}{\dtwo}{\ltwo}$.
  \end{lemma}
  \begin{proof}
    By induction on types.
  \end{proof}
\subsection{A Fundamental Lemma}
Usually, the main result about any system of logical relations is the
so-called \emph{Fundamental Lemma}, which states that any typable term
is in relation \emph{with itself}. But how would the Fundamental Lemma
look like here? Should any term be at \emph{somehow minimal distance}
to itself, in the spirit of what happens, e.g. with metrics \cite{Pierce/DistanceMakesTypesGrowStronger/2010,GaboardiEtAl/POPL/2017}?
Actually, there is no hope to prove anything like that for
differential logical relations, as the following example shows.
\begin{example}
  Consider again the term $\lid=\abs{\vone}{\vone}$, which
  can be given type $\tone=\reals\arr\reals$ in the empty context.
  Please recall that $\metdom{\tone}=\rset\times\pirset\arr\pirset$.
  Could we prove that $\delta_\tone(\lid,0_\tone,\lid)$,
  where $0_\tone$ is the constant-$0$ function? The answer is negative:
  given two real numbers $\rone$ and $\rtwo$ at distance $\varepsilon$,
  the terms $\lid\rone$ and $\lid\rtwo$ are themselves $\varepsilon$ apart,
  thus at nonnull distance. The best one can say, then, is that
  $\delta_\tone(\lid,f,\lid)$, where $f(x,\varepsilon)=\varepsilon$.
\end{example}
As the previous example suggests, a term $\lone$ being at self-distance $d$ 
is a witness of $\lone$ being \emph{sensitive to changes} to the environment
according to $d$. Indeed, the only terms which are at self-distance $0$
are the constant functions. This makes the underlying theory more general
than the one of logical or metric relations, although the latter can be
proved to be captured by differential logical relations, as we will see
in the next section.

Coming back to the question with which we opened the section, 
we can formulate a suitable fundamental lemma for 
differential logical relations
 by stating that for any closed term 
$\lone$ of type $\tone$ there exists 
$d\in\metdom{\tone}$ such that $(\lone,d,\lone)\in\delta_{\tone}$. 
In order to prove such result, however, we need to prove something stronger, 
namely the extension of the above statement to arbitrary (and thus possibly open) 
terms. Doing so, requires to extend differential logical relations to 
arbitrary sequents $\Gamma \vdash \tone$.

Let us begin extending the maps $\sem{\cdot}$ and 
$\metdom{\cdot}$ to environments. 

\begin{definition}
Given an environment $\eone$, define: 
$$
\sem{\eone} = \prod_{(x:\tone) \in \eone} \sem{\tone};
\qquad
\metdom{\eone} = \prod_{(x:\tone) \in \eone} \metdom{\tone}.
$$
\end{definition}
An element in e.g. $\metdom{\eone}$ is thus a family 
$\alpha \in \prod_{(x:\tone) \in \eone}$, meaning that 
for any $(x: \tone) \in \eone$, $\alpha(x) \in \metdom{\tone}$. 
The syntactic counterparts of such families is given by families 
$\vlfone \in \prod_{(x: \tone) \in \eone} \cv{\tone}$. 
We refer to $\vlfone$ as a \emph{$\eone$-family of values}. Indeed,
such a $\eone$-family of values can naturally be seen as a substitution 
mapping each variable $(x: \tone) \in \Gamma$ to $\vlfone(x) \in\cv{\tone}$. 
As it is customary, for a term $\Gamma \vdash \lone: \tone$ 
we write $\lone \vlfone$ for the
closed term of type $\tone$ obtained applying the substitution 
$\vlfone$ to $\lone$. We denote by
$\cv{\eone}$ the set of all $\eone$-family of values.

We can now extend a differential logical relation $\{\delta_{\tone}\}_{\tone}$ 
to environments 
stipulating that the family $\{\delta_{\eone} \subseteq \cv{\eone} 
\times \metdom{\eone} \times \cv{\eone}\}_{\eone}$ 
is a differential logical relation if
$$
\delta_{\eone}(\vlfone, \alpha, \vlftwo) 
\iff \forall (x: \tone) \in \eone.\ 
\delta_{\tone}(\vlfone(x), \alpha(x), \vlftwo(x)).
$$

Next, we extend our framework to arbitrary sequents $\eone \vdash \tone$. 
First of all, we define:
$$
\sem{\eone \vdash \tone} = \sem{\tone}^{\sem{\eone}}
\qquad
\metdom{\eone \vdash \tone} = \metdom{\tone}^{\sem{\eone} \times \metdom{\eone}}.
$$

It is then natural to extend $\{\delta_{\tone}\}_{\tone}$ to arbitrary 
sequents by stipulating that $\{\delta_{\eone \vdash \tone} \subseteq 
\Lambda_{\eone \vdash \tone} \times \metdom{\eone \vdash \tone} \times 
\Lambda_{\eone \vdash \tone}\}_{\eone \vdash \tone}$, where 
$\Lambda_{\eone \vdash \tone}$ denotes the set of sequents typable 
within the sequent $\eone \vdash \tone$, is a differential logical relation 
if $\delta_{\eone \vdash \tone}(\lone, d, \ltwo)$ holds if and only if
 $$
  \diffmet{\eone}{\vlfone}{\alpha}{\vlftwo}
  \implies
  \diffmet{\tone}{\lone\vlfone}{\done(\sem{\vlfone},\alpha)}{\ltwo \vlftwo}.
$$
This definition as it is, however, does not work well, as it does not take into account 
possible alternations of the substitutions $\vlfone$ and 
$\vlftwo$. To solve this issue we introduce the following notation. 
Given a boolean-valued map $B \in \{0,1\}^{\eone}$, 
and two substitutions $\vlfone, \vlftwo \in \cv{\eone}$, define 
the substitution $\select{\vlfone}{\vlftwo}{B} \in \cv{\eone}$ as:
$$
  \select{\vlfone}{\vlftwo}{B}(x) =\left\{
  \begin{array}{ll}
    \vlfone(x) & \mbox{if $b(x) =0$}\\
    \vlftwo(x) & \mbox{if $b(x) =1$}.\\
  \end{array}\right.
  $$
We can now extend differential logical relations to arbitrary terms.

\begin{definition}
Given a differential logical relations $\{\delta_{\tone}\}_{\tone}$, 
define $\{\delta_{\eone \vdash \tone} \subseteq 
\Lambda_{\eone \vdash \tone} \times \metdom{\eone \vdash \tone} \times 
\Lambda_{\eone \vdash \tone}\}_{\eone \vdash \tone}$ stipulating 
that $\delta_{\eone \vdash \tone}(\lone, d, \ltwo)$ holds if and only if
 $$
  \diffmet{\eone}{\vlfone}{\alpha}{\vlftwo}
  \implies \forall B \in \{0,1\}^{\eone}
  \diffmet{\tone}{\lone\select{\vlfone}{\vlftwo}{B}}{\done(\sem{\vlfone},\alpha)}
  {\ltwo\select{\vlftwo}{\vlfone}{B}}.
$$
\end{definition}

Before proving the desired strengthening of the fundamental lemma, 
it is useful to observe the following useful result.

\begin{lemma}
\label{lemma:auxialiary-lemma-fundamental-lemma}
For all terms $\cdot \vdash \lone, \ltwo: \tone$, we have 
$\delta_{\tone}(\lone, d, \ltwo) \iff \delta_{\tone}(NF(\lone), d, NF(\ltwo))$.
\end{lemma}

\begin{proof}
First observe that if $\delta_{\tone}(\lone', d, \ltwo)$ and 
$\lone \to_{\beta} \lone'$ (where $\to_{\beta}$ is the obvious 
small-step semantics relation associated to $\Downarrow$), 
then $\delta_{\tone}(\lone, d, \ltwo)$ holds (a similar statement 
holds for $\ltwo$). By Theorem
 ~\ref{thm:norm} we thus infer the right-to-left 
implication. For the other implication, we proceed by induction 
on $\tone$ observing that $NF(NF(\lone)V) = NF(\lone V)$.
\end{proof}

\begin{lemma}
For any $\eone \vdash \lone: \tone$, there exists 
$d \in \metdom{\eone \vdash \tone}$ such that 
$\delta_{\eone \vdash \tone}(\tone, d, \tone)$ holds.
\end{lemma}

\begin{proof}
The proof is by induction on $\eone \vdash \lone: \tone$.
\begin{itemize}
  \item Suppose $\eone \vdash x: \tone$, meaning that $x \in \eone$. 
    We simply define $d: \sem{\eone} \times \metdom{\eone} \to \metdom{\tone}$ 
    as $d(g,\alpha) = \alpha(x)$. 
  \item Suppose $\eone \vdash r: \reals$. Define 
    $d: \sem{\eone} \times \metdom{\eone} \to \mathbb{R}_{\geq 0}^{\infty}$ 
    as $d(g,\alpha) = 0$.
  \item Suppose $\eone \vdash f: \reals^n \to \reals$. For simplicity, we show the 
    case for $n = 1$ (the case for $n > 1$ follows the same structure). We have to find 
    $d \in \metdom{\reals \to \reals}^{\sem{\eone} \times \metdom{\eone}}$ such 
    that $\delta_{\eone}(\vlfone, \alpha, \vlftwo)$ implies 
    $\delta_{\reals \to \reals}(f, d(\sem{\vlfone}, \alpha), f)$. 
    The latter means, that for all values $\cdot \vdash V, W \vdash \reals$  
    (for simplicity, we denote by $V$ both the numeral $V$ and the number 
    $\sem{V}$)
    and $e \in \mathbb{R}_{\geq 0}^{\infty}$, 
    $|V - W| \leq e$ implies 
    $|NF(f(V)) - NF(f(W))| \leq d(\sem{\vlfone}, \alpha)(V, e)$. 
    Define:
    $$
    A = \bigcup_{e' \in \mathbb{R}, e' \leq e} 
    [V - e', V + e'];
    \qquad
    d(\sem{\vlfone}, \alpha)(\sem{V}, e) = diam(f(A)),
    $$
    where the diameter $diam(X)$ of a set $X \subseteq \mathbb{R}$ is defined 
    as $\sup_{x,y \in X} | x - y|$. Notice that a set can have diameter 
    $\infty$. We conclude the wished thesis observing that 
    $|V - W| \leq e$ implies $V, W \in A$ 
    (recall that $V,W$ being real numbers, so is $|V - W|$), 
    and thus $f(V), f(W) \in f(A)$. As a consequence, $|f(V) - f(W)| \leq 
    diam(f(A))$, and thus we are done.
  \item Suppose \(\vcenter{\infer{\eone \vdash \lone\ltwo: \tone}
    {\eone \vdash \lone: \ttwo \to 
    \tone & \eone \vdash \ltwo: \ttwo} }\). By induction hypothesis 
     we have:
     \begin{enumerate}
      \item There exists 
        $d_1 \in \metdom{\ttwo \to \tone}^{\sem{\eone} \times \metdom{\eone}}$ 
        such that
        $\delta_{\eone}(\vlfone, \alpha, \vltwo)$
        implies for any $B \in \{0,1\}^{\eone}$,  
        for all $V,W \in CV(\ttwo)$, and for any $x \in \metdom{\ttwo}$:
        \[
        \delta_{\ttwo}(V,x,W) 
        \implies \delta_{\tone}(M\select{\vlfone}{\vlftwo}{B}V, 
        d_1(\sem{\vlfone}, \alpha)(\sem{V},x), M\select{\vlftwo}{\vlfone}{B}W).
        \]
      \item There exists $d_2 \in \metdom{\ttwo}^{\sem{\eone} \times \metdom{\eone}}$ 
        such that $\delta_{\eone}(\vlfone, \alpha, \vltwo)$
        implies for any $B \in \{0,1\}^{\eone}$:
        \[
        \delta_{\ttwo}(N\select{\vlfone}{\vlftwo}{B}, 
        d_2(\sem{\vlfone}, \alpha), N\select{\vlftwo}{\vlfone}{B}W).
        \]
      \end{enumerate}
   We define the wished $d$ as 
  $d(\sem{\vlfone}, \alpha) = 
  \sup_{V} d_1(\sem{\vlfone}, \alpha)(\sem{V},d_2(\sem{\vlfone}, \alpha))$. 
  We conclude the wished thesis instantiating 
  $V$ as $NF(N\select{\vlfone}{\vlftwo}{B})$ relying on 
  Lemma \ref{lemma:auxialiary-lemma-fundamental-lemma}.
\item Suppose \(\vcenter{\infer{\eone \vdash \lambda x.\lone: \tone \to \ttwo}
    {\eone, x:\tone \vdash \lone: \ttwo}}\). The thesis directly follows 
    from the induction hypothesis.
\end{itemize}

\end{proof}

As an immediate corollary we obtain the wished result.

\begin{theorem}[Fundamental Lemma, Version I]\label{thm:fundamental}
  For every $\tj{\cdot}{\lone}{\tone}$ there is a 
  $d\in\metdom{\tone}$ such that $(\lone,d,\lone)\in\delta_{\tone}$.
\end{theorem}
But what do we gain from Theorem~\ref{thm:fundamental}? In the classic
theory of logical relations, the Fundamental Lemma has, as an easy corollary,
that logical relations are compatible: it suffices to invoke the theorem
with any context $C$ \emph{seen as a term} $C[\vone]$, such that
$\tj{\vone:\tone,\eone}{C[\vone]}{\ttwo}$. Thus, ultimately, logical
relations are proved to be a \emph{compositional} methodology for
program equivalence, in the following sense: if $\lone$ and $\ltwo$ are
equivalent, then $C[\lone]$ and $C[\ltwo]$ are equivalent, too.

In the realm of differential logical relations, the Fundamental Lemma
plays a similar role, although with a different, \emph{quantitative}
flavor: once $C$ has been proved sensitive to changes according to $\done$,
and $\vlone,\vltwo$ are proved to be at distance $\dtwo$, then, e.g.,
the impact of substituting $\vlone$ with $\vltwo$ in $C$ can be
measured by composing $\done$ and $\dtwo$ (and $\sem{\vlone}$),
i.e. by computing $\done(\sem{\vlone},\dtwo)$. Notice that the sensitivity
analysis on $C$ and the relational analysis on $\vlone$ and $\vltwo$
are decoupled. What the Fundamental Lemma tells you is that $\done$
and $\dtwo$ can \emph{always} be found.

\subsection{Our Running Example, Revisited}
It is now time to revisit the example we talked about in the
Introduction.  Consider the following two programs, both closed and of
type $\reals\arr\reals$:
$$
\lsin=\abs{\vone}{\sine(\vone)};\qquad\lid=\abs{\vone}{\vone}.
$$
Let us prove that $(\lsin,f,\lid)\in\delta_{\reals\arr\reals}$, where
$f(x,y)=y+\av{x-\sin x}$. Consider any pair of real numbers
$\rone,\rtwo\in\rset$ such that $|\rone-\rtwo|\leq\varepsilon$, where
$\varepsilon\in\pirset$.  We have that:
\begin{align*}
  \av{\sin r-s}&\;=\;\av{\sin r-r+r-s}\leq\av{\sin r-r}+\av{r-s}\leq\av{\sin r-r}+\varepsilon=f(r,\varepsilon)\\
  \av{\sin s-r}&\;=\;\av{\sin s-\sin r+\sin r-r}\leq\av{\sin s-\sin r}+\av{\sin r-r}\leq\av{s-r}+\av{\sin r-r}\\
  &\;\leq\;\varepsilon+\av{\sin r-r}=f(r,\varepsilon).
\end{align*}
The fact that $\av{\sin s-\sin r}\leq\av{s-r}$ is a consequence
of $\sin$ being $1$-Lipschitz continuous.
This can be proved, e.g., by way of Lagrange Theorem:
\begin{theorem}[Lagrange]
	Let $f:[a,b]\rightarrow\mathbb{R}$ a function continuous in $[a,b]$ and differentiable in $(a,b)$. Then there exists a point $c\in(a,b)$ such that
	$$
	\frac{f(b)-f(a)}{b-a}=f'(c).
	$$
\end{theorem}
The following is an easy corollary
\begin{proposition}
	For each $\vartheta,\varphi$, $|\sin(\vartheta)-\sin(\varphi)|\leq|\vartheta-\varphi|$.
\end{proposition}
\begin{proof}
  If $\vartheta=\varphi$ the result is trivial. Then without loss of
  generality we consider $\varphi\ <\vartheta$. We apply Lagrange
  Theorem to the interval $[\varphi,\vartheta]$. In particular, there
  exists a point $\psi\in(\varphi,\vartheta)$ such that
  $$
  \frac{\sin(\vartheta)-\sin(\varphi)}{\vartheta-\phi}=\sin'(\psi)=\cos(\psi).
  $$
  Since $-1\leq\cos(\psi)\leq 1$,
  $$
  \left|\frac{\sin(\vartheta)-\sin(\varphi)}{\vartheta-\phi}\right|\leq 1
  $$
  from which follows $|\sin(\vartheta)-\sin(\varphi)|\leq|\vartheta-\varphi|$.
\end{proof}  

Now, consider a context $\cone$ which makes use of either $\lsin$
or $\lid$ by feeding them with a value close to $0$, call it
$\theta$. Such a context could be, e.g., $\cone=(\abs{\vone}{\vone(\vone\theta)})\chole$.
$\cone$ can be seen as a term having type $\tone=(\reals\arr\reals)\arr\reals$.
A self-distance $\done$ for $\cone$ can thus be defined as an
element of
$$
\metdom{\tone}=\sem{\reals\arr\reals}\times\metdom{\reals\arr\reals}\arr\pirset.
$$
namely $F=\metalambda\pair{g}{h}.h(g(\theta),h(\theta,0))$. This allows
for compositional reasoning about program distances: the overall impact
of replacing $\lsin$ by $\lid$ can be evaluated by computing
$F(\sem{\lsin},f)$. Of course the context $\cone$ needs to be taken into
account, but \emph{once and for all}: the functional $F$
can be built without any access to either $\lsin$ or $\lid$.

\section{Logical and Metric Relations as DLRs}
The previous section should have convinced the reader about the
peculiar characteristics of differential logical relations compared to 
(standard) metric
and logical relations. In this section we show that despite the
apparent differences, logical and metric relations can somehow
be retrieved as specific kinds of program differences. This is, however,
bound to be nontrivial. The na\"ive attempt, namely seeing program
equivalence as being captured by \emph{minimal} distances in logical
relations, fails: the distance between a program \emph{and itself}
can be nonnull.

How should we proceed, then? Isolating those distances which witness
program equivalence is indeed possible, but requires a bit of an
effort.  In particular, the sets of those distances can be, again,
defined by induction on $\tone$. For every $\tone$, we give
$\nullset{\tone} \subseteq \metdom{\tone}$ by induction on the
structure of $\tone$:
$$
\nullset{\reals} = \{0\}\qquad\qquad
\nullset{\tone\prd\ttwo} = \nullset{\tone} \prd \nullset{\ttwo}
$$
$$
\nullset{\tone\arr\ttwo} = \{\fone\mid
    \forall \elone \in \sem{\tone}. \forall \vtwo \in \nullset{\tone}.
    \fone(\elone,\vtwo) \in \nullset{\ttwo}\}
$$
\begin{lemma}\label{lem:nullcomplete}
  $\nullset{\tone}$ is complete.
  That is, $\forall \setone \subseteq \nullset{\tone}. \sup{\setone} \in \nullset{\tone}$.
\end{lemma}
\begin{proof}
  By induction on types.
  \begin{varitemize}
    \item Case $\reals$.
      Obvious: $\sup{\emptyset} = \sup{\{0\}} = 0 \in \nullset{\reals}$.
    \item Case $\tone\arr\ttwo$.
      Let $\setone \subseteq \nullset{\tone\arr\ttwo}$.
      Note that
      $\forall\fone\in\setone. \forall \vone \in \cv{\tone}. \forall\done\in\nullset{\tone}.
      \fone(\vone,\done)\in\nullset{\ttwo}$.
      The function $s = \sup\setone \colon \cv{\tone}\prd\metdom{\tone}\arr\metdom{\ttwo}$
      is defined by
      $s(\vone,\done) = \sup\{\fone(\vone,\done) \,|\, \fone \in \setone\}$.
      We verify that
      $\forall\vone\in\cv{\tone}. \forall\done\in\nullset{\tone}.
      s(\vone,\done)\in\nullset{\ttwo}$.
      This is true since given $\vone\in\cv{\tone}$ and $\done\in\nullset{\tone}$,
      $\fone(\vone,\done)$ is in $\nullset{\ttwo}$, and thus
      $s(\vone,\done) = \sup\{\fone(\vone,\done) \,|\, \fone\in\setone\} \in \nullset{\ttwo}$
      by I.H.
    \item Let $\setone\subseteq\nullset{\tone\prd\ttwo}$. Then
      $\sup\setone = (\sup(\prjleft\setone), \sup(\prjright\setone)) \in \nullset{\tone}\prd\nullset{\ttwo} = \nullset{\tone\prd\ttwo}$
      by I.H.
  \end{varitemize}
\end{proof}

\begin{lemma}\label{lem:nullupwardclosed}
  Let $\done, \done' \in \nullset{\tone}$ and $\done \leq_{\tone} \done'$.
  Let $\{\delta_{\ttwo}\}_{\ttwo}$ be a differential logical relation.
  If $\diffmet{\tone}{\lone}{\done}{\ltwo}$
  then $\diffmet{\tone}{\lone}{\done'}{\ltwo}$.
\end{lemma}
\begin{proof}
  By induction on types.
  \begin{varitemize}
    \item Case $\reals$.
      Obvious, since $\nullset{\reals} = \{0\}$.
    \item Case $\tone\arr\ttwo$.
      Let $\diffmet{\tone\arr\ttwo}{\lone}{\fone}{\ltwo}$ and $\fone \leq \fone'$.
      Then $\forall \vone \in \cv{\tone}. \forall \done \in \metdom{\tone}. \forall \vone \in \cv {\tone}.$
      \[
      \begin{array}{lll}
        \diffmet{\tone}{\vone}{\done}{\vtwo} &\Rightarrow& \diffmet{\ttwo}{\lone\vone}{\fone(\vone,\done)}{\ltwo\vtwo} \land \diffmet{\ttwo}{\lone\vtwo}{\fone(\vone,\done)}{\ltwo\vone}\\
        &\Rightarrow& \diffmet{\ttwo}{\lone\vone}{\fone'(\vone,\done)}{\ltwo\vtwo} \land \diffmet{\ttwo}{\lone\vtwo}{\fone'(\vone,\done)}{\ltwo\vone}\\
        && \text{by $\fone(\vone,\done) \leq \fone'(\vone,\done)$ and I.H.}
      \end{array}
      \]
      Thus $\diffmet{\tone\arr\ttwo}{\lone}{\fone'}{\ltwo}$.
    \item Case $\tone\prd\ttwo$.
      Let $\diffmet{\tone\prd\ttwo}{\lone}{(\done_1, \done_2)}{\ltwo}$ and
      $(\done_1, \done_2) \leq (\done'_1, \done'_2)$.
      By definition $\diffmet{\tone}{\prjleft\lone}{\done_1}{\prjleft\ltwo}$,
      $\diffmet{\ttwo}{\prjright\lone}{\done_2}{\prjright\ltwo}$,
      $\done_1\leq\done'_1$, and $\done_2\leq\done'_2$.
      Thus by I.H.\ $\diffmet{\tone}{\prjleft\lone}{\done'_1}{\prjleft\ltwo}$ and
      $\diffmet{\ttwo}{\prjright\lone}{\done'_2}{\prjright\ltwo}$ hold,
      hence $\diffmet{\tone\prd\ttwo}{\lone}{(\done'_1, \done'_2)}{\ltwo}$.
  \end{varitemize}
\end{proof}

Notice that $\nullset{\tone\arr\ttwo}$ is not defined as 
$\sem{\tone} \times \nullset{\tone} \arr \nullset{\ttwo}$ (doing so would violate 
$\nullset{\tone\arr\ttwo} \subseteq \metdom{\tone \arr \ttwo}$). 
The following requires some effort, and testifies that, indeed, program
equivalence in the sense of logical relations precisely corresponds to
being at a distance in $\nullset{\tone}$:
\begin{theorem}\label{thm:logdiff}
  Let $\{\mathcal{L}_{\tone}\}_{\tone}$ be a logical relation.  There
  exists a differential logical relation $\{\delta_{\tone}\}_\tone$
  satisfying $\logrel{\tone}{\lone}{\ltwo} \Longleftrightarrow \exists
  \done \in \nullset{\tone}. \diffmet{\tone}{\lone}{\done}{\ltwo}$.
\end{theorem}
\begin{proof}
  By induction on types.
  \begin{varitemize}
    \item Case $\reals$.
      Define $\delta_{\reals}$ by
      $\diffmet{\reals}{\lone}{0}{\ltwo}$ 
      if and only if $\logrel{\tone}{\lone}{\ltwo}$ or $\logrel{\tone}{\ltwo}{\lone}$.
    \item Case $\tone\arr\ttwo$.
      
      \noindent $\underline{\Leftarrow}$.
      Assume that there exists $\fone \in \nullset{\tone\arr\ttwo}$ satisfying
      $\diffmet{\tone\arr\ttwo}{\lone}{\fone}{\ltwo}$.
      Then for all $\lone', \ltwo'$ of type $\tone$,
      \[
      \begin{array}{lrl}
        \logrel{\tone}{\lone'}{\ltwo'} & \stackrel{\text{I.H.}}{\Leftrightarrow} & \exists \done' \in \nullset{\tone}. \diffmet{\tone}{\lone'}{\done'}{\ltwo'}\\ 
        & \Rightarrow & \diffmet{\ttwo}{\lone\lone'}{\fone(\lone',\done')}{\ltwo\ltwo'}\\
        & \stackrel{\text{I.H.}}{\Leftrightarrow} & \logrel{\ttwo}{\lone\lone'}{\ltwo\ltwo'}.
      \end{array}
      \]
      Thus $\logrel{\tone\arr\ttwo}{\lone}{\ltwo}$ holds.

      \noindent $\underline{\Rightarrow}$.
      Assume $\logrel{\tone\arr\ttwo}{\lone}{\ltwo}$.
      Define a function $\fmn\colon\cv{\tone}\prd\metdom{\tone}\arr\metdom{\ttwo}$ by:
      \[
      \fmn(\lone',\done) =
      \sup \bigcup_{\ltwo'\colon\diffmet{\tone}{\lone'}{\done}{\ltwo'}}
      \{\done' \in \nullset{\ttwo}\midd \diffmet{\ttwo}{\lone\lone'}{\done'}{\ltwo\ltwo'}\}.
      \]
      We show that
      \begin{varitemize}
        \item{(I)} $\fmn \in \nullset{\tone\arr\ttwo}$
        \item{(II)} $\diffmet{\tone\arr\ttwo}{\lone}{\fmn}{\ltwo}$.
      \end{varitemize}
      (I) is immediate since the right-hand side of the definition of $\fmn(\lone',\done)$
      is always contained by $\nullset{\ttwo}$ by Lemma~\ref{lem:nullcomplete}.
      (II) is equivalent to the following by definition
      for each $\lone' \in \cv{\tone}$ and $\done \in \metdom{\tone}$:
      \[
      \forall \ltwo' \in \cv{\ttwo}. 
      \diffmet{\tone}{\lone'}{\done}{\ltwo'}
      \Rightarrow \diffmet{\ttwo}{\lone\lone'}{\fmn(\lone,\done)}{\ltwo\ltwo'}
      \land \diffmet{\ttwo}{\lone\ltwo'}{\fmn(\lone,\done)}{\ltwo\lone'}.
      \]
      Fix $\ltwo'$ arbitrarily.
      If $\diffmet{\tone}{\lone'}{\done}{\ltwo'}$ does not hold, the implication vacuously holds.
      Assume $\diffmet{\tone}{\lone'}{\done}{\ltwo'}$ holds.
      Then we have
      \[
      \begin{array}{lrl}
        \diffmet{\tone}{\lone'}{\done}{\ltwo'}
        & \stackrel{\text{I.H.}}{\Leftrightarrow} & \logrel{\tone}{\lone'}{\ltwo'}\\ 
        & \Rightarrow & \logrel{\ttwo}{\lone\lone'}{\ltwo\ltwo'}\\
        & \stackrel{\text{I.H.}}{\Leftrightarrow} & \exists\done' \in \nullset{\ttwo}. \diffmet{\ttwo}{\lone\lone'}{\done'}{\ltwo\ltwo'}.
      \end{array}
      \]
      By definition $\done' \leq_{\ttwo} \fmn(\lone',\done)$ for such $\done' \in \nullset{\ttwo}$.
      Thus $\diffmet{\ttwo}{\lone\lone'}{\fmn(\lone',\done)}{\ltwo\ltwo'}$ also holds
      by Lemma~\ref{lem:nullupwardclosed}.
      Similarly $\diffmet{\ttwo}{\lone\ltwo'}{\fmn(\lone',\done)}{\ltwo\lone'}$.
      Hence $\diffmet{\tone\arr\ttwo}{\lone}{\fmn}{\ltwo}$.

    \item Case $\tone\prd\ttwo$.
      Let $\logrel{\tone\prd\ttwo}{\lone}{\ltwo}$.
      By definition of logical relations,
      $\logrel{\tone}{\prjleft\lone}{\prjleft\ltwo}$
      and $\logrel{\ttwo}{\prjright\lone}{\prjright\ltwo}$ hold.
      By I.H.\ there exist $\done_1 \in \nullset{\tone}$ and $\done_2 \in \nullset{\ttwo}$
      that satisfy $\diffmet{\tone}{\prjleft\lone}{\done_1}{\prjleft\ltwo}$
      and $\diffmet{\ttwo}{\prjright\lone}{\done_2}{\prjright\ltwo}$;
      thus $\diffmet{\tone\prd\ttwo}{\lone}{(\done_1,\done_2)}{\ltwo}$ holds
      for $(\done_1,\done_2) \in \nullset{\tone\prd\ttwo}$.
      The other direction also holds by I.H.
  \end{varitemize}
\end{proof}

What if we want to generalise the argument above to metric relations,
as introduced, e.g., by Reed and
Pierce~\cite{Pierce/DistanceMakesTypesGrowStronger/2010}.  The set
$\nullset{\tone}$ becomes a set of distances parametrised by a single
real number:
$$
\metset{\reals}{\rone} = \{\rone\}\qquad\qquad
\metset{\tone\prd\ttwo}{\rone} = \metset{\tone}{\rone} \prd \metset{\ttwo}{\rone}
$$
$$
\metset{\tone\arr\ttwo}{\rone} = \{\fone\mid
    \forall \elone \in \sem{\tone}. \forall \vtwo \in \metset{\tone}{\rtwo}.
    \fone(\elone,\vtwo) \in \metset{\ttwo}{\rone+\rtwo}\}
$$

\section{Strengthening the Fundamental Theorem through Finite Distances}
Let us now ask ourselves the following question: given any term
$\lone\in\ct{\tone}$, what can we say about its sensitivity, i.e.,
about the values $\done\in\metdom{\tone}$ such that
$\delta_\tone(\lone,\done,\lone)$?  Two of the results we have proved
about $\STlamreal$ indeed give partial answers to the aforementioned
question. On the one hand, Theorem~\ref{thm:fundamental} states that
such a $\done$ can \emph{always} be found. On the other hand,
Theorem~\ref{thm:logdiff} tells us that such a $\done$ can be
taken in $\nullset{\tone}$. Both these answers are not particularly
informative, however. The mere existence of such a
$\done\in\metdom{\tone}$, for example, is trivial since $\done$ can
always be taken as $\dinf$, the maximal element of the underlying
quantale. The fact that such a $\done$ can be taken
from $\nullset{\tone}$ tells us that, e.g. when
$\tone=\ttwo\arr\tthree$, $\lone$ returns equivalent terms when fed
with equivalent arguments: there is no quantitative guarantee about the
behaviour of the term when fed with non-equivalent arguments.

Is this the best one can get about the sensitivity of $\STlamreal$ terms?
The absence of full recursion suggests that we could hope to prove that
infinite distances, although part of the underlying quantale, can in
fact be useless. In other words, we are implicitly suggesting that
self-distances could be elements of $\findom{\tone}\subset\metdom{\tone}$,
defined as follows:
$$
\findom{\reals}=\rset_{\geq 0}; \qquad\qquad \findom{\tone\prd\ttwo} = \findom{\tone}\prd\findom{\ttwo};
$$
$$
\findom{\tone\arr\ttwo}=
  \{\fone \in \metdom{\tone\arr\ttwo}
  \;\mid\; \forall x \in\sem{\tone}.\forall\elemone\in\findom{\tone}.
  \fone(x,\elemone) \in \findom{\ttwo}\}.\\
$$
\begin{wrapfigure}{R}{.28\textwidth}
  \fbox{
    \begin{minipage}{.26\textwidth}
      \centering
    \includegraphics[scale=1.0]{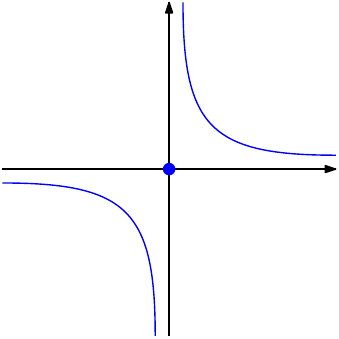}
  \end{minipage}}
  \caption{A total, but highly discontinuous, function.}\label{fig:noncontinuous}
\end{wrapfigure}
Please observe that $\findom{\tone}$ is in general a much larger set
of differences than $\bigcup_{\rone\in\pirset}\metset{\tone}{\rone}$: the
former equals the latter only when $\tone$ is $\reals$. Already when
$\tone$ is $\reals\arr\reals$, the former includes, say, functions
like $f(\rone,\varepsilon)=(\rone+\varepsilon)^2$, while the latter
does not.

Unfortunately, there are terms in $\STlamreal$ which cannot be proved
to be at self-distance in $\findom{\tone}$, and, surprisingly, this is
\emph{not} due to the higher-order features of $\STlamreal$, but to
$\{\rfuncs{n}\}_{n\in\nset}$ being arbitrary, and containing functions
which do not map finite distances to finite distances, like
$$
h(\rone)=\left\{
\begin{array}{ll}
  0               & \mbox{if $\rone=0$} \\
  \frac{1}{\rone} & \mbox{otherwise}
\end{array}
\right.
$$
(see Figure~\ref{fig:noncontinuous}). Is this phenomenon \emph{solely}
responsible for the necessity of finite self-distances in $\STlamreal$?
The answer is positive, and the rest of this section is devoted precisely
to formalising and proving the aforementioned conjecture.

First of all, we need to appropriately axiomatise the absence of unbounded
discontinuities from $\{\rfuncs{n}\}_{n\in\nset}$. A not-so-restrictive but
sufficient axiom turns out to be weak boundedness: a function
$\fone_n:\rset^n\arr\rset$ is said to be \emph{weakly bounded} if and only
if it maps bounded subsets of $\rset^n$ into bounded subsets of $\rset$.
As an example, the function $h$ above is \emph{not} weakly bounded, because
$h([-\varepsilon,\varepsilon])$ is
$$
\left(-\infty,-\frac{1}{\varepsilon}\right]\cup\{0\}\cup\left[\frac{1}{\varepsilon},\infty\right)
$$
which is unbounded for any as $\varepsilon>0$. Any term $\lone$ is said to
be weakly bounded iff any function symbol $\fone_n$ occurring in $\lone$ is
itself weakly bounded. Actually, this is precisely what one needs to get the
strengthening of the Fundamental Theorem we are looking for.
  \begin{lemma}\label{lem:evalpreserve}
    Let $\tj{\cdot}{\lone,\ltwo}{\tone}$ and $\eval{\lone}{\vlone}$.
    Then $\diffmetobsolete{\tone}{\lone}{\ltwo} = \diffmetobsolete{\tone}{\vlone}{\ltwo}$.
  \end{lemma}
  \begin{proof}
    By induction on type $\tone$; do case analysis for each type constructor.
  \end{proof}

  \begin{lemma}\label{lem:downwardclosed}
    Let $\tone$ be a type and $\elemone \in \findom{\tone}$.
    If $\elemone' \leq_{\tone} \elemone$ then $\elemone' \in \findom{\tone}$.
  \end{lemma}
  \begin{proof}
    By induction on type $\tone$. Straightforward.
  \end{proof}
\begin{theorem}[Fundamental Theorem, Version II]\label{thm:fundamental-II}
  For any weakly bounded term $\tj{\cdot}{\lone}{\tone}$, there is
  $d\in\findom{\tone}$ such that $(\lone,d,\lone)\in\delta_{\tone}$.
\end{theorem}
\begin{proof}
  By induction on derivation of $\tj{\eone}{\lone}{\tone}$.
  In order not to make the proof too unreadable,
  we prove it for the case $n=1$, meaning $\eone = \vfour\colon\ttwo$.
  It is straightforward (but tedious) to extend it to an arbitrarily large environment.
  \begin{varitemize}
  \item Case
    \[
    \infer{\tj{\eone}{\vfour}{\ttwo}}{\vfour:\ttwo\in\eone}.
    \]
    Then
    \[
    \begin{array}{ll} 
      \odiffmetobsolete{\eone}{\ttwo}{\vfour}{\vfour}
      &= \diffmetobsolete{\ttwo\arr\ttwo}{\lambda\vfour.\vfour}{\lambda\vfour.\vfour}\\
      &= \metabs{\langle\vone^{\cv{\ttwo}},\vtwo^{\metdom{\ttwo}}\rangle}
                       {\sup_{\vthree\colon\diffmetobsolete{\ttwo}{\vthree}{\vone} \leq \vtwo}
                         \{\diffmetobsolete{\ttwo}{(\abs{\vfour}{\vfour}) \vone}{(\abs{\vfour}{\vfour}) \vthree},
                         \diffmetobsolete{\ttwo}{(\abs{\vfour}{\vfour}) \vthree}{(\abs{\vfour}{\vfour}) \vone}\}}\\
        &= \metabs{\langle\vone^{\cv{\ttwo}},\vtwo^{\metdom{\ttwo}}\rangle}
                       {\sup_{\vthree\colon\diffmetobsolete{\ttwo}{\vthree}{\vone} \leq \vtwo}
                         \{\diffmetobsolete{\ttwo}{ \vone}{ \vthree},
                   \diffmetobsolete{\ttwo}{ \vthree}{ \vone}\}}
                 \text{ by Lemma~\ref{lem:evalpreserve}}\\
        &= \metabs{\langle\vone^{\cv{\ttwo}},\vtwo^{\metdom{\ttwo}}\rangle}
                 {\sup_{\vthree\colon\diffmetobsolete{\ttwo}{\vthree}{\vone} \leq \vtwo}
                   \diffmetobsolete{\ttwo}{ \vthree}{ \vone}}
                 \text{ by Lemma~\ref{lem:symmetric}}\\
        &\leq\metabs{\langle\vone^{\cv{\ttwo}},\vtwo^{\metdom{\ttwo}}\rangle}{\vtwo}
                 \in \findom{\ttwo\arr\ttwo}.
      \end{array}
      \]
      Hence $\odiffmetobsolete{\eone}{\ttwo}{\vfour}{\vfour} \in \findom{\ttwo\arr\ttwo}$
      by Lemma~\ref{lem:downwardclosed}.

    \item Case          
      \[
      \infer{\tj{\eone}{\rone}{\reals}}{}.
      \]
      Easy.
    \item Case
      \[
      \infer{\tj{\eone}{\pfone{n}}{\reals^n\arr\reals}}{}.
      \]    
      Straightforward.
    \item Case
      \[
      \infer{\tj{\eone}{\abs{\vone}{\lone}}{\tone\arr\ttwo}}
            {\tj{\eone,\vone:\tone}{\lone}{\ttwo}}.
      \]
      By induction hypothesis.
    \item Case
      \[
      \infer{\tj{\eone}{\lone\ltwo}{\tone_{2}}}
            {\tj{\eone}{\lone}{\tone_{1}\arr\tone_{2}} & \tj{\eone}{\ltwo}{\tone_{1}}}.
      \]
      Induction hypothesis in this case can be written down as:
      \[
      \begin{array}{l}
        \forall \vone_{1} \in \cv{\ttwo}.
        \forall \vtwo_{1} \in \findom{\ttwo}.
        \forall \vone_{2} \in \cv{\tone_{1}}.
        \forall \vtwo_{2} \in \findom{\tone_{1}}.\\
        \sup_{\vthree_{1}\colon\diffmetobsolete{\ttwo}{\vthree_{1}}{\vone_{1}} \leq \vtwo_{1}}
        \sup_{\vthree_{2}\colon\diffmetobsolete{\ttwo}{\vthree_{2}}{\vone_{2}} \leq \vtwo_{2}}
        \{\diffmetobsolete{\tone_{2}}
                  {\sbs{\lone}{\vone_{1}}{\vfour}\vone_{2}}
                  {\sbs{\lone}{\vthree_{1}}{\vfour}\vthree_{2}},
        \diffmetobsolete{\tone_{2}}
                {\sbs{\lone}{\vone_{1}}{\vfour}\vthree_{2}}
                {\sbs{\lone}{\vthree_{1}}{\vfour}\vone_{2}}
        \}
        \in \findom{\tone_{2}}\\
        \cdots\text{(I)}
      \end{array}
      \]
      \[
      \begin{array}{l}
        \forall \vone_{1} \in \cv{\ttwo}.
        \forall \vtwo_{1} \in \findom{\ttwo}.\\
        \sup_{\vthree_{1}\colon\diffmetobsolete{\ttwo}{\vthree_{1}}{\vone_{1}} \leq \vtwo_{1}}
        \diffmetobsolete{\tone_{2}}
                {\sbs{\ltwo}{\vone_{1}}{\vfour}}
                {\sbs{\ltwo}{\vthree_{1}}{\vfour}}
        \in \findom{\tone_{1}}
      \end{array} \cdots\text{(II)}
      \]
      We show that
      $\odiffmetobsolete{\eone}{\tone_{2}}{\lone\ltwo}{\lone\ltwo} \in \findom{\tone_{2}}$,
      i.e.\ $\diffmetobsolete{\tone_{2}}{\abs{\vfour}{\lone\ltwo}}{\abs{\vfour}{\lone\ltwo}} \in \findom{\ttwo \arr \tone_{2}}$.
      By definition we must show that
      \[
      \begin{array}{l}
        \forall \vone_{1} \in \cv{\ttwo}.
        \forall \vtwo_{1} \in \findom{\ttwo}.\\
        \sup_{\vthree\colon\diffmetobsolete{\ttwo}{\vthree}{\vone_{1}} \leq \vtwo_{1}}
        \{\diffmetobsolete{\tone_{2}}{(\abs{\vfour}{\lone\ltwo})\vone_{1}}{(\abs{\vfour}{\lone\ltwo})\vthree},
        \diffmetobsolete{\tone_{2}}{(\abs{\vfour}{\lone\ltwo})\vthree}{(\abs{\vfour}{\lone\ltwo})\vone_{1}}
        \}\\
        = \sup_{\vthree\colon\diffmetobsolete{\ttwo}{\vthree}{\vone_{1}} \leq \vtwo_{1}}
        \diffmetobsolete{\tone_{2}}{(\abs{\vfour}{\lone\ltwo})\vthree}{(\abs{\vfour}{\lone\ltwo})\vone_{1}}
        \text{ by Lemma~\ref{lem:symmetric}}\\
        = \sup_{\vthree\colon\diffmetobsolete{\ttwo}{\vthree}{\vone_{1}} \leq \vtwo_{1}}
        \diffmetobsolete{\tone_{2}}
                {\sbs{\lone}{\vthree}{\vfour}\sbs{\ltwo}{\vthree}{\vfour}}
                {\sbs{\lone}{\vone_{1}}{\vfour}\sbs{\ltwo}{\vone_{1}}{\vfour}}
        \text{ by Lemma~\ref{lem:evalpreserve}}\\
        \in \findom{\tone_{2}}.
      \end{array}
      \]
      By definition of $\sup{}$ the following inequality holds
      for all $\vone_{1} \in \cv{\ttwo}$ and $\vtwo_{1} \in \findom{\ttwo}$:
      \[
      \begin{array}{l}
        \sup_{\vthree_{1}\colon\diffmetobsolete{\ttwo}{\vthree_{1}}{\vone_{1}} \leq \vtwo_{1}}
        \diffmetobsolete{\tone_{2}}
                {\sbs{\lone}{\vthree_{1}}{\vfour}\sbs{\ltwo}{\vthree_{1}}{\vfour}}
                {\sbs{\lone}{\vone_{1}}{\vfour}\sbs{\ltwo}{\vone_{1}}{\vfour}}\\
        \leq
        \sup_{\vthree_{1}\colon\diffmetobsolete{\ttwo}{\vthree_{1}}{\vone_{1}} \leq \vtwo_{1}}
        \{\diffmetobsolete{\tone_{2}}
                  {\sbs{\lone}{\vone_{1}}{\vfour}\sbs{\ltwo}{\vone_{1}}{\vfour}}
                  {\sbs{\lone}{\vthree_{1}}{\vfour}\sbs{\ltwo}{\vthree_{1}}{\vfour}},
          \diffmetobsolete{\tone_{2}}
                  {\sbs{\lone}{\vone_{1}}{\vfour}\sbs{\ltwo}{\vthree_{1}}{\vfour}}
                  {\sbs{\lone}{\vthree_{1}}{\vfour}\sbs{\ltwo}{\vone_{1}}{\vfour}}
        \}\\
        \leq
        \sup_{\vthree_{1}\colon\diffmetobsolete{\ttwo}{\vthree_{1}}{\vone_{1}} \leq \vtwo_{1}}
        \sup_{\vthree_{2}\colon\diffmetobsolete{\ttwo}{\vthree_{1}}{\sbs{\ltwo}{\vone_{1}}{\vfour}} \leq s}\\
        \quad \{\diffmetobsolete{\tone_{2}}
                  {\sbs{\lone}{\vone_{1}}{\vfour}\sbs{\ltwo}{\vone_{1}}{\vfour}}
                  {\sbs{\lone}{\vthree_{1}}{\vfour}\vthree_{2}},
          \diffmetobsolete{\tone_{2}}
                  {\sbs{\lone}{\vone_{1}}{\vfour}\vthree_{2}}
                  {\sbs{\lone}{\vthree_{1}}{\vfour}\sbs{\ltwo}{\vone_{1}}{\vfour}}
        \},
      \end{array}
      \]
      where
      $s = \sup_{\vthree_{1}\colon\diffmetobsolete{\ttwo}{\vthree_{1}}{\vone_{1}} \leq \vtwo_{1}}
      \diffmetobsolete{\tone_{2}}
              {\sbs{\ltwo}{\vone_{1}}{\vfour}}
              {\sbs{\ltwo}{\vthree}{\vfour}}.$

      Thus by Lemma~\ref{lem:downwardclosed} it suffices show that
      \[
      \begin{array}{l}
      \sup_{\vthree_{1}\colon\diffmetobsolete{\ttwo}{\vthree_{1}}{\vone_{1}} \leq \vtwo_{1}}
      \,\sup_{\vthree_{2}\colon\diffmetobsolete{\ttwo}{\vthree_{1}}{\sbs{\ltwo}{\vone_{1}}{\vfour}} \leq s}\\
      \{\diffmetobsolete{\tone_{2}}
                {\sbs{\lone}{\vone_{1}}{\vfour}\sbs{\ltwo}{\vone_{1}}{\vfour}}
                {\sbs{\lone}{\vthree_{1}}{\vfour}\vthree_{2}},
        \diffmetobsolete{\tone_{2}}
                {\sbs{\lone}{\vone_{1}}{\vfour}\vthree_{2}}
                {\sbs{\lone}{\vthree_{1}}{\vfour}\sbs{\ltwo}{\vone_{1}}{\vfour}}\}
     \in \findom{\tone_{2}}.
     \end{array}
     \]
     By I.H.\ (I), this holds if $\sbs{\ltwo}{\vone_{1}}{\vfour} \in \cv{\tone_{2}}$
     and $s = \sup_{\vthree_{1}\colon\diffmetobsolete{\ttwo}{\vthree_{1}}{\vone_{1}} \leq \vtwo_{1}}
     \diffmetobsolete{\tone_{2}}{\sbs{\ltwo}{\vone_{1}}{\vfour}}{\sbs{\ltwo}{\vthree}{\vfour}}
     \in \findom{\tone_{2}}$;
     the latter holds by I.H.\ (II).
    \item Case
      \[
      \infer{\tj{\eone}{\pair{\lone}{\ltwo}}{\tone\prd\ttwo}}
            {\tj{\eone}{\lone}{\tone} & \tj{\eone}{\ltwo}{\ttwo}}.
      \]
      Straightforward using induction hypothesis.
    \item Case
      \[
      \infer{\tj{\eone}{\prjleft}{\tone\prd\ttwo\arr\tone}}{}.
      \]
      Straightforward.
    \item Case
      \[
      \infer{\tj{\eone}{\prjright}{\tone\prd\ttwo\arr\ttwo}}{}.
      \]      
      Straightforward, similar to the case of $\prjleft$.
    \end{varitemize}
    
  \end{proof}

The reader may have wondered about how restrictive a condition weak
boundedness really is. In particular, whether it corresponds to some
form of continuity. In fact, the introduced condition only rules out
unbounded discontinuities. In other words, weak boundedness can be
equivalently defined by imposing local boundedness \emph{at any
  point} in the domain $\rset$. This is weaker than asking
for boundedness, which requires the existence of a global bound.
\section{A Categorical Perspective}\label{sect:categorical}
Up to now, differential logical relations have been treated very concretely,
without looking at them through the lens of category theory. This is
in contrast to, e.g., the treatment of metric relations from~\cite{GaboardiEtAl/POPL/2017},
in which soundness of metric relations for \FUZZ\ is obtained as a
byproduct of a proof of symmetric monoidal closedness for the category
$\mathbf{MET}$ of pseudometric spaces and Lipschitz functions.

But what could take the place of pseudometric spaces in a categorical
framework capturing differential logical relations? The notion of
a metric needs to be relaxed along at least two axes. On the one
hand, the codomain of the ``metric'' $\delta$ is not necessarily
the set of real numbers, but a more general structure, namely
a quantale. On the other, as we already noticed, it is not necessarily
true that equality implies indistancy, but rather than indistancy
implies inequality. What comes out of these observations is,
quite naturally, the notion of a generalized metric domain, itself
a generalisation of partial metrics~\cite{DBLP:journals/tamm/BukatinKMP09}. 
The rest of this section
is devoted to proving that the category of generalised metric domains
is indeed cartesian closed, thus forming a model of simply typed
$\lambda$-calculi.

Formally, given a quantale $\qtlone = (Q, \le_Q, 0_{Q}, \mult{Q})$\footnote{
When unambiguous, we will omit subscripts in $\le_Q$, $0_{Q}$, and $\mult{Q}$.}, a 
\emph{generalised metric domain} on $\qtlone$ is
a pair $(\setone,\delta_\setone)$, where $\setone$ is a set and
$\delta_\setone$ is a subset of $\setone\times\qtlone\times\setone$ satisfying
some axioms akin to those of a metric domain:
\begin{align*}
  \diffmet{\setone}{\elone}{0_{Q}}{\eltwo}
  &\Rightarrow\elone=\eltwo; 
  && 
  \mbox{(Indistancy Implies Equality)}\\
  \diffmet{\setone}{\elone}{\done}{\eltwo}
  &\Rightarrow\diffmet{\setone}{\eltwo}{\done}{\elone};
  && \mbox{(Symmetry)}\\
  \diffmet{\setone}{\elone}{\done}{\eltwo}
  \wedge\diffmet{\setone}{\eltwo}{\dtwo}{\eltwo}
  \wedge\diffmet{\setone}{\eltwo}{\dthree}{\elthree}
  &\Rightarrow\diffmet{\setone}{\elone}{\done \mult{} \dtwo \mult{} \dthree}{\elthree}. 
  && \mbox{(Triangularity)}
\end{align*}
Please observe that $\delta_\setone$ is a \emph{relation} rather than a function. Moreover,
the first axiom is dual to the one typically found in, say, pseudometrics. The third
axiom, instead, resembles the usual triangle inequality for pseudometrics, 
but with the crucial difference that since objects
can have non-null self-distance, such 
a distance has to be taken into account. Requiring equality to imply indistancy 
(and thus $\diffmet{\setone}{\elone}{0_{Q}}{\eltwo}
\Leftrightarrow\elone=\eltwo$), we see that \mbox{(Triangularity)} gives exactly the usual
triangle inequality (properly generalised to quantale and relations 
\cite{Hoffman-Seal-Tholem/monoidal-topology/2014,Lawvere/GeneralizedMetricSpaces/1973}).

In this section we show that generalised metric domains form a
cartesian closed category, unlike that of metric spaces (which is
known to be non-cartesian closed).  As a consequence, we obtain a firm
categorical basis of differential logical relations.  The category of
generalised metric domain, denoted by $\catgmd$.
\begin{definition}
  The category $\catgmd$ has the following data.
  \begin{varitemize}
  \item An object $\objone$ is a triple $(A, \qtlone, \delta)$
    where $\qtlone$ is a quantale and 
    $(A, \delta)$ is a generalized metric domain on $\qtlone$.
  \item An arrow $(A, \qtlone, \delta) \to (B, \qtltwo, \rho)$
    is a pair $(\arrone,\qarrone)$ consisting of a function $\arrone\colon A \to B$ and another function $\qarrone\colon Q \times A \to S$ satisfying
    $\forall a,a' \in A. \forall q \in Q.
    \delta(a,q,a') \Rightarrow \rho(\arrone(a), \qarrone(q,a), \arrone(a'))$ and $\rho(\arrone(a), \qarrone(q,a'), \arrone(a'))$.
  \end{varitemize}
\end{definition}
We can indeed give $\catgmd$ the structure of a category. In fact, 
the identity on the object $\objone = (A, \qtlone, \delta)$ in $\catgmd$
is given by $(\id_{\objone}, \id_{\objone}')$
where $\id_{\objone}\colon A \to A$ is the set-theoretic identity on $A$
and $\id_{\objone}'\colon Q \times A \to Q$ is defined by
$\id_{\objone}'(q,a) = q$.
The composition of two arrows
$(\arrone,\qarrone)\colon (A, \qtlone, \delta) \to (B, \qtltwo, \rho)$
and $(\arrtwo,\qarrtwo)\colon (B, \qtltwo, \rho) \to (C, \qtlthree, \nu)$
is the pair $(\arrthree, \qarrthree)$ where
$\arrthree\colon A \to C$ is given by the function composition $\arrtwo \circ \arrone \colon A \to C$ and $\arrthree\colon Q \times A \to T$ is given by $\qarrthree(q,a) = \qarrtwo(\qarrone(q,a),\arrone(a))$. Straightforward calculations show that 
composition is associative, and that the identity arrow behaves as its neutral 
element.

\begin{lemma}
  \label{lemma:GMD-category}
  $\catgmd$ is a category.
\end{lemma}

\begin{proof}
  The identity on the object $\objone = (A, \qtlone, \delta)$ in $\catgmd$
  is given by $(\id_{\objone}, \id_{\objone}')$
  where $\id_{\objone}\colon A \to A$ is the set-theoretic ideneity on $A$
  and $\id_{\objone}'\colon Q \times A \to Q$ is defined by
  $\id_{\objone}'(q,a) = q$.
  The composition of two arrows
  $(\arrone,\qarrone)\colon (A, \qtlone, \delta) \to (B, \qtltwo, \rho)$
  and $(\arrtwo,\qarrtwo)\colon (B, \qtltwo, \rho) \to (C, \qtlthree, \nu)$
  is the pair $(\arrthree, \qarrthree)$ where
  $\arrthree\colon A \to C$ is given by the function composition $\arrtwo \circ \arrone \colon A \to C$ and $\arrthree\colon Q \times A \to T$ is given by $\qarrthree(q,a) = \qarrtwo(\qarrone(q,a),\arrone(a))$.
  The fact that every identity arrow and every composition are again arrows in $\catgmd$ can be checked straightforwardly:
  $\id_{\objone}$ clearly satisfies
  $\relone(a,q,a') \Rightarrow \relone(\id(a),\id'(q,a),\id(a'))$
  and $\relone(\id(a),\id'(q,a'),id(a'))$
  since the right-hand side is by definition $\relone(a,q,a)$.
  We have to check that $(\arrthree, \qarrthree)$ satisfies
  $\relone(a,q,a') \Rightarrow \relthree(\arrthree(a), \qarrthree(q,a), \arrthree(a'))$ (and its symmetric part).
  Indeed, we have $\relone(\arrone(a),\qarrone(q,a), \arrone(a'))$
  since $(\arrone, \qarrone)$ is an arrow in $\catgmd$;
  and $\relthree(\arrtwo(\arrone(a)), \qarrtwo(\qarrone(q,a),\arrone(a)), \arrtwo(\arrone(a')))$, which is equivalent to $\relthree(\arrthree(a), \qarrthree(q,a), \arrthree(a'))$ by definition, holds since $(\arrtwo, \qarrtwo)$ is an arrow in $\catgmd$. (Similarly for the symmetric part.)
\end{proof}

The construction of (finite) products in $\catgmd$ is mostly straightforward:
\begin{lemma}
  \label{lemma:GMD-products}
  The category $\catgmd$ has a terminal object and binary products.
\end{lemma}
\begin{proof}
  The terminal object $\catgmd$ is defined as 
  $(\{*\}, \mathbb{O}, \delta_0)$, where $\mathbb{O}$ is the 
  one-element quantale $\{0\}$, and $\delta_0 = \{(*, 0, *)\}$. 
  Clearly, $(\{*\}, \delta_0)$ is a generaliszed metric domain on 
  $\mathbb{O}$ which behaves as a terminal object.
  Next, we show that $\catgmd$ has binary products.
  Given two objects $\objone$ and $\objtwo$,
  their binary product $\objone\times\objtwo$ is given by a triple
  $(A \times B, \qtlone\times\qtltwo, \delta\times\rho)$.
  The projection from $\objone\times\objtwo$ to $\objone$ is the pair $(\prjleft, \prjleft')$ given by $\prjleft(a,b) = a$ and $\prjleft'((q,s),(a,b)) = q$ (similarly for $\prjright$).

  Given objects $\objone, \objtwo, \objthree$
  and arrows $(\arrone,\qarrone)\colon\objthree\to\objone, (\arrtwo,\qarrtwo)\colon\objthree\to\objtwo$,
  the mediating arrow $\objthree\to\objone\times\objtwo$ is given by a
  pair $(\arrthree, \qarrthree)$ where
  $\arrthree\colon\setthree\to\setone\times\settwo$ is a function
  $\arrthree(c) = (\arrone(c),\arrtwo(c))$
  and
  $\qarrthree\colon\qsetthree\times\setthree\to\qsetone\times\qsettwo$ is a function
  $\qarrthree(t,c) = (\qarrone(t,c),\qarrtwo(t,c))$.
  This satisfies
  $(\prjleft, \prjleft')\after(\arrthree, \qarrthree) = (\arrone, \qarrone)$:
  on the first component,
  $(\prjleft\after\arrthree)(c)
  = \prjleft(\arrtwo(c), \arrthree(c))
  = \arrtwo(c)$.
  On the second,
  $(\prjleft'\after\qarrthree)(t,c)
  = \prjleft'(\qarrthree(t,c), \arrthree(c))
  = \prjleft'((\qarrone(t,c), \qarrtwo(t,c)) , \arrthree(c))
  = \qarrone(t,c)$
  (similarly $(\prjright, \prjright')\after(\arrthree, \qarrthree) = (\arrtwo, \qarrtwo)$ holds).
\end{proof}

\begin{lemma}
  \label{lemma:GMD-exponenetials}
  The category $\catgmd$ has exponential objects.
\end{lemma}
\begin{proof}
  The exponential object $\objthree^{\objtwo}$ is given by a triple
  $(\setthree^{\settwo}, \qtlthree^{\qtltwo\times\settwo}, \relthree^{\reltwo})$
  where $\setthree^{\settwo}$ is the function space $\{f \,|\, f\colon \settwo \to \setthree\}$,
  $\qtlthree^{\qtltwo\times\settwo}$ is the exponential quantale,
  and $\relthree^{\reltwo}$ is a ternary relation over
  $\setthree^{\settwo} \times T^{S \times \settwo}\times \setthree^{\settwo}$
  defined by:
  if $\reltwo(b,s,b')$ then $\relthree(f(b), d(s,b), f'(b'))$
  and $\relthree(\arrone(b), d(s,b'), \qarrone(b'))$.
  The relation $\relthree^{\reltwo}$ is indeed a differential logical relation.

  Then given three objects $\objone, \objtwo, \objthree$ and
  an arrow $\arrtwo \colon \objone \times \objtwo \to \objthree$,
  the abstraction $\lambda \arrtwo \colon \objone \to \objthree^\objtwo$ of $\arrtwo$
  is given by a pair $(\lambda \arrtwo, \lambda \qarrtwo)$ where
  $\lambda \arrtwo\colon\setone \to \setthree^\settwo$ is defined by
  $\lambda \arrtwo(a)(b) = \arrtwo(a,b)$ and
  $\lambda \qarrtwo\colon \qsetone \times \setone \to \qsetthree^{\qsettwo\times\settwo}$ is defined by
  $\lambda \qarrtwo(q,a)(s,b) = \qarrtwo((q,s),(a,b))$.

  Last, $\evalarr \colon \objthree^\objtwo\times\objtwo\to\objthree$
  is given by a pair $(\evalarr,\evalarr')$ where
  $\evalarr\colon\setthree^\settwo\times\settwo\to\setthree$ is defined by
  $\evalarr(f,b) = f(b)$
  and $\evalarr'\colon(\qsetthree^{\qsettwo\times\settwo}\times\qsettwo)\times(\setthree^\settwo\times\settwo)\to\qsetthree$
  is defined by 
  $\evalarr'((d,s),(f,b)) = d(s,b)$.

  The fact that these arrows makes the diagram
  \[
  \xymatrix{
    \objone\times\objtwo \ar@{.>}[d]_{\lambda \arrtwo \times \id} \ar[dr]^{\arrtwo}&
    \\
    \objthree^{\objtwo}\times\objtwo \ar[r]^{\evalarr} & \objthree
  }
  \]
  commute can be checked componentwise.
  Let $(\arrthree, \qarrthree)$ be the pair comprising the composition
  $\evalarr \after (\lambda \arrtwo \times \id) \colon \objone \times \objtwo \to \objthree$.
  On the first component, we have $\arrthree((a,b))
  = \evalarr((\lambda \arrtwo \times \id)(a,b))
  = \evalarr((\lambda \arrtwo(a),b))
  = \lambda \arrtwo(a)(b)
  = \arrtwo(a,b)$.
  On the second,
  $\qarrthree((q,s),(a,b))
  = \evalarr'((\lambda\qarrtwo \times \id')(q,s)(a,b),
  (\lambda\arrtwo\times\id)(a,b))
  = \evalarr'((\lambda\qarrtwo(q,a), \id'(s,b)),
  (\lambda\arrtwo(a), b))
  = \evalarr'((\lambda\qarrtwo(q,a), s),
  (\lambda\arrtwo(a), b))
  = \lambda\qarrtwo(q,a)(s,b)
  = \qarrtwo((q,s),(a,b))$.

\end{proof}

Joining together \ref{lemma:GMD-products} and
and \ref{lemma:GMD-exponenetials} we obtain the wished result.

\begin{corollary}
The category $\catgmd$ is cartesian closed.
\end{corollary}

Interestingly, the constructions of product and exponential objects
closely match the definition of a differential logical relation. In
other words, differential logical relations as given in
Definition~\ref{def:dlr} can be seen as providing a denotational model
of $\STlamreal$ in which base types are interpreted by the generalised
metric domain corresponding to the Euclidean distance.

\section{Conclusion}
In this paper, we introduced differential logical relations as a novel methodology
to evaluate the ``distance'' between programs of higher-order calculi akin to the
$\lambda$-calculus. We have been strongly inspired by some unpublished work by
Westbrook and Chaudhuri~\cite{WestbrookAndChaudhuri}, 
who were the first to realise that evaluating differences
between interactive programs requires going beyond mere real numbers. We indeed
borrowed our running examples from the aforementioned work.

This paper's contribution, then consists in giving a simple definition of differential
logical relations, together with some results about their underlying metatheory: two
formulations of the Fundamental Lemma, a result relating differential logical relations
and ordinary logical relations, a categorical framework in which generalised metric
domains --- the metric structure corresponding to differential logical relations --- are
proved to form a cartesian closed category. Such results give evidence
that, besides being \emph{more expressive} than metric relations, differential logical
relations are somehow \emph{more canonical}, naturally forming a model of simply-typed
$\lambda$-calculi.

As the title of this paper suggests, we see the contributions above just as a very
first step towards understanding the nature of differences in a logical environment.
In particular, at least two directions deserve to be further explored.
\begin{varitemize}
\item
  The first one concerns \emph{language features}: admittedly, the calculus
  $\STlamreal$ we consider here is very poor in terms of its
  expressive power, lacking full higher-order recursion and thus not
  being universal. Moreover, $\STlamreal$ does not feature any form of
  effect, including probabilistic choices, in which evaluating
  differences between programs would be very helpful. 
  Addressing such issues seems to require to impose a domain 
  structure on generalised metric domain, on one hand, and to look at monads 
  on $\catgmd$, on the other hand (for the latter, the literature on monadic 
  lifting for quantale-valued relations might serve as a guide 
  \cite{Hoffman-Seal-Tholem/monoidal-topology/2014}).
\item
  The second one is about \emph{abstract differences}: defining differences
  as functions with \emph{the same rank} as that of the compared programs implies
  that reasoning about them is complex. Abstracting differences so as to
  facilitate differential reasoning could be the way out, given that deep connections
  exist between logical relations and abstract interpretation~\cite{Abramsky1990}.
\end{varitemize}

\bibliography{main}

\end{document}

%% file: macros.tex
\usepackage{proof}
\usepackage{amsfonts,amsmath}
\usepackage{amsmath}
\usepackage{amssymb}
\usepackage{stmaryrd}
\usepackage{graphicx}
\usepackage[usenames, dvipsnames]{xcolor}
\usepackage[all]{xy}
\usepackage{url}

\newcommand{\bnf}{::=}
\newcommand{\midd}{\; \; \mbox{\Large{$\mid$}}\;\;}

\newcommand{\tone}{\tau}
\newcommand{\ttwo}{\rho}
\newcommand{\tthree}{\xi}
\newcommand{\reals}{\mathit{REAL}}
\newcommand{\arr}{\rightarrow}

\newcommand{\prd}{\times}
\newcommand{\eone}{\Gamma}
\newcommand{\etwo}{\Delta}
\newcommand{\emenv}{\emptyset}
\newcommand{\tj}[3]{#1\vdash #2:#3}

\newcommand{\abs}[2]{\lambda #1.#2}
\newcommand{\pair}[2]{\langle #1,#2\rangle}
\newcommand{\tuple}[1]{\langle #1\rangle}
\newcommand{\tpl}[2]{\langle #1_{1},\ldots,#1_{#2}\rangle}
\newcommand{\vone}{x} 
\newcommand{\vtwo}{y}
\newcommand{\vthree}{z}
\newcommand{\vfour}{w}
\newcommand{\lone}{M} 
\newcommand{\ltwo}{N}
\newcommand{\lthree}{L}
\newcommand{\lfour}{P}
\newcommand{\lid}{M_{\mathit{ID}}} 
\newcommand{\lsin}{M_{\mathit{SIN}}} 
\newcommand{\lset}[1]{\Lambda_{#1}} 
\newcommand{\rone}{r} 
\newcommand{\rtwo}{s}
\newcommand{\vset}{\mathbb{V}}
\newcommand{\nset}{\mathbb{N}}
\newcommand{\pfone}[1]{f_{#1}}
\newcommand{\prede}{\mathit{pred}_1}
\newcommand{\sine}{\mathit{sin}_1}
\newcommand{\prjleft}{\pi_1}
\newcommand{\prjright}{\pi_2}
\newcommand{\iflz}[2]{\texttt{iflz}\;#1\;\texttt{else}\;#2}
\newcommand{\iter}[2]{\texttt{iter}\;#1\;\texttt{base}\;#2}
\newcommand{\vlone}{V} 
\newcommand{\vltwo}{W}
\newcommand{\vlfone}{\mathbf{V}} 
\newcommand{\vlftwo}{\mathbf{W}}
\newcommand{\eval}[2]{#1\Downarrow #2}
\newcommand{\terminate}[1]{#1\Downarrow}
\newcommand{\sbs}[3]{#1\{#2/#3\}}      
\newcommand{\cv}[1]{\mathit{CV}({#1})} 
\newcommand{\ct}[1]{\mathit{CT}({#1})} 
\newcommand{\nf}[1]{\mathit{NF}({#1})} 
\newcommand{\redset}[1]{\mathit{RED}_{#1}}
\newcommand{\vredset}[1]{\mathit{VRED}_{#1}}

\newcommand{\STlamreal}{\mathit{ST}^{\lambda}_{\mathbb{R}}}

\newcommand{\cone}{C} 

\newcommand{\chole}{[\cdot]}

\newcommand{\select}[3]{{#3}^{#1}_{#2}}

\newcommand{\metalambda}{%
  \mathop{%
    \rlap{$\lambda$}%
    \mkern2mu
    \raisebox{.275ex}{$\lambda$}%
  }%
}
\newcommand{\metabs}[2]{\metalambda {#1}.{#2}}
\newcommand{\fone}{f}
\newcommand{\rset}{\mathbb{R}}

\newcommand{\pirset}{\mathbb{R}_{\geq 0}^\infty}
\newcommand{\ftwo}{g}
\newcommand{\rfuncs}[1]{\mathcal{F}_{#1}}
\newcommand{\av}[1]{\left|#1\right|}

\newcommand{\sem}[1]{\llbracket #1\rrbracket}
\newcommand{\logrel}[3]{\mathcal{L}_{#1}({#2},{#3})} 
\newcommand{\metdom}[1]{\llparenthesis{#1}\rrparenthesis}

\newcommand{\findom}[1]{\metdom{#1}^{<\infty}}
\newcommand{\diffmet}[4]{\delta_{#1}({#2},{#3},{#4})} 
\newcommand{\done}{d} 
\newcommand{\dtwo}{e} 
\newcommand{\dthree}{f}

\newcommand{\dinf}{d_\infty}
\newcommand{\elemone}{t} 
\newcommand{\elemtwo}{u}
\newcommand{\elemthree}{s}
\newcommand{\elemfour}{r}
\newcommand{\setone}{A} 
\newcommand{\settwo}{B} 
\newcommand{\setthree}{C} 
\newcommand{\elone}{x} 
\newcommand{\eltwo}{y}
\newcommand{\elthree}{z}

\newcommand{\relone}{\delta}
\newcommand{\reltwo}{\rho}
\newcommand{\relthree}{\nu}
\newcommand{\qtlone}{\mathbb{Q}}
\newcommand{\qtltwo}{\mathbb{S}}
\newcommand{\qtlthree}{\mathbb{T}}
\newcommand{\qsetone}{{Q}}
\newcommand{\qsettwo}{{S}}
\newcommand{\qsetthree}{{T}}
\newcommand{\mult}[1]{*_{#1}} 
\newcommand{\indset}{I} 
\newcommand{\indone}{i} 

\newcommand{\unit}[1]{e_{#1}} 
\newcommand{\nullset}[1]{\metdom{#1}^0}
\newcommand{\metset}[2]{\metdom{#1}^{#2}} 
\newcommand{\fmn}{\fone^{\lone,\ltwo}}

\newcommand{\catgmd}{\mathbf{GMD}}
\newcommand{\objone}{\mathcal{A}}
\newcommand{\objtwo}{\mathcal{B}}
\newcommand{\objthree}{\mathcal{C}}
\newcommand{\arrone}{f}
\newcommand{\arrtwo}{g}
\newcommand{\arrthree}{h}
\newcommand{\qarrone}{\zeta}
\newcommand{\qarrtwo}{\eta}
\newcommand{\qarrthree}{\theta}
\newcommand{\after}{\mathbin{\circ}}
\newcommand{\id}{\mathrm{id}}

\newcommand{\evalarr}{\mathrm{eval}}

\newcommand{\diffmetobsolete}[3]{\delta^{#1}({#2},{#3})} 
\newcommand{\odiffmetobsolete}[4]{\delta_{#1}^{#2}({#3},{#4})} 

\newenvironment{varitemize}
{
\begin{list}{\labelitemi}
{\setlength{\itemsep}{0pt}
 \setlength{\topsep}{0pt}
 \setlength{\parsep}{0pt}
 \setlength{\partopsep}{0pt}
 \setlength{\leftmargin}{15pt}
 \setlength{\rightmargin}{0pt}
 \setlength{\itemindent}{0pt}
 \setlength{\labelsep}{5pt}
 \setlength{\labelwidth}{10pt}
}}
{
 \end{list} 
}
\newcounter{number}

\newcounter{glc}
\newtheorem{theorem}[glc]{Thoerem}
\newtheorem{corollary}[glc]{Corollary}
\newtheorem{lemma}[glc]{Lemma}
\newtheorem{proposition}[glc]{Proposition}
\newtheorem{definition}[glc]{Definition}
\newtheorem{example}[glc]{Example}